\documentclass[a4paper]{article}
\usepackage[usenames,dvipsnames]{xcolor}
\usepackage{amsfonts}
\usepackage{amsmath,amsthm,amssymb,dsfont,pifont}
\usepackage{dsfont}
\usepackage{enumerate}
\usepackage[english]{babel}
\usepackage{graphicx}	
\usepackage{subcaption}
\usepackage{geometry}
\usepackage{url}
\usepackage{todonotes}
\usepackage{bbm}
\usepackage{caption}
\usepackage{subcaption}
\usepackage{cite}
\usepackage{physics}

\usepackage{epsfig}
\usetikzlibrary{shapes.symbols,patterns} % for source symbols
\usepackage{hyperref}[breaklinks]
\hypersetup{colorlinks=true,citecolor=blue,linkcolor=blue,filecolor=blue,urlcolor=blue}

\usepackage{nicefrac}
\usepackage{mathtools}

\usepackage{algorithm}
\usepackage{algpseudocode}

\usepackage[capitalize]{cleveref}
\usepackage{mathrsfs}

\usepackage{thm-restate}

\theoremstyle{plain}
  \newtheorem{theorem}{Theorem}[section]
  \newtheorem{lemma}[theorem]{Lemma}  
  \newtheorem{result}{Result}
  \newtheorem{fact}[theorem]{Fact}
  
  \newtheorem{claim}[theorem]{Claim}
  \newtheorem{corollary}[theorem]{Corollary}
  
  \theoremstyle{definition}
  \newtheorem{definition}[theorem]{Definition}

  \theoremstyle{remark}
  \newtheorem{remark}[theorem]{Remark}

\newcommand*{\Id}{\mathrm{Id}}

\newcommand*{\R}{\mathbb{R}}

\newcommand*{\eps}{\varepsilon}

\newcommand*{\poly}{\mathrm{poly}}

\newcommand*{\spec}{\mathrm{spec}}

\newcommand*{\opnorm}{\mathrm{op}}
\newcommand*{\trnorm}{\mathrm{tr}}

\usepackage{comment}
\usepackage{thm-restate}

\usepackage{xcolor}
\definecolor{famBrown}{RGB}{153,76,0} % un marrone scuro

\title{Nearly optimal algorithms to learn sparse quantum Hamiltonians in physically motivated distances}

\author{ Amira Abbas\\ \texttt{\scriptsize Google Quantum AI} \and
Nunzia Cerrato\\ \texttt{\scriptsize Scuola Normale Superiore} \and
Francisco Escudero Gutiérrez\\ \texttt{\scriptsize Centrum Wiskunde \& Informatica (CWI), QuSoft}\\ \texttt{\scriptsize and Quriosity, Inria-Saclay}\\ \and 
Dmitry Grinko\\ \texttt{\scriptsize ILLC, University of Amsterdam}\\ \texttt{\scriptsize and QuSoft  } \and 
Francesco Anna Mele\\ \texttt{\scriptsize NEST, Scuola Normale Superiore}\\ \texttt{\scriptsize and Istituto Nanoscienze}\and
Pulkit Sinha\\ \texttt{\scriptsize Institute for Quantum Computing,} \\ \texttt{\scriptsize University of Waterloo} }
\date{}
\allowdisplaybreaks

\begin{document}
\maketitle
\begin{abstract}
    In this work, we study the problem of learning quantum Hamiltonians $H$ that are $s$-sparse in the Pauli basis, given access to their time-evolution operators. Although Hamiltonian learning has been extensively investigated, two issues recur in much of the existing literature: the absence of lower bounds establishing optimality and the use of mathematically convenient but physically opaque error measures.
    
    We address both challenges by introducing two physically motivated notions of distance between Hamiltonians and designing a nearly optimal algorithm with respect to one of these metrics. The first, the \emph{time-constrained distance}, quantifies distinguishability through dynamical evolution up to a bounded time. The second, the \emph{temperature-constrained distance}, captures distinguishability through thermal states at bounded inverse temperatures.

    We show that $s$-sparse Hamiltonians with bounded operator norm can be learned under both distances using only $\widetilde O(s \log(1/\varepsilon))$ experiments and $\widetilde O(s^2/\varepsilon)$ total evolution time. For the time-constrained distance, we further establish lower bounds of $\Omega((s/n)\log(1/\varepsilon) + s)$ experiments and $\Omega(\sqrt{s}/\varepsilon)$ total evolution time, demonstrating near-optimality in the number of experiments.

    As an intermediate result, we obtain an algorithm that learns every Pauli coefficient of $s$-sparse Hamiltonians up to error $\eps$ in $\widetilde O(s \log(1/\varepsilon))$ experiments and $\widetilde O(s/\varepsilon)$ total evolution time, improving upon several recent results. 
    
    The source of this improvement is a new \emph{isolation technique}, inspired by the Valiant-Vazirani theorem (STOC’85), which shows that \emph{NP is as easy as detecting unique solutions}. This isolation technique allows us to query the time evolution of a single Pauli coefficient of a sparse Hamiltonian—even when the Pauli support of the Hamiltonian is unknown—ultimately enabling us to recover the Pauli support itself.
    
\end{abstract}
\newpage
\tableofcontents
\section{Introduction}
With the rapid development of quantum hardware, the design of protocols to characterise its dynamics and its behaviour at thermal equilibrium has become increasingly more important \cite{bravyi2024high,liu2025robust}. Both aspects are ultimately governed by the system Hamiltonian, which has motivated an extensive literature on Hamiltonian learning \cite{Silva2011Practical,holzapfel2015scalable,Zubida2021Optimal,haah2022optimal,yu2023robust,Dutkiewicz.2023,huang2023heisenberg,li2023heisenberglimited,franca2024efficient,Gu2022Practical,zhao2025learning,hu2025ansatz,anshu2021sample,onorati2023efficient,rouze2023learning,bakshi2023learning,ma2024learning,arunachalam2025testing,caro2023learning,sinha2025improved}. However, two recurring shortcomings limit the practical relevance of many results. First, lower bounds establishing optimality are often missing. Second, many works measure the quality of estimation with mathematically convenient distances that lack direct physical interpretation, in terms of a worst-case error. This raises the following question:\\
\begin{quote}
\center \emph{What is the optimal complexity of Hamiltonian learning\\ in physically meaningful worst-case distances?}
\end{quote}\vspace{0.4cm}

We address this question by presenting a nearly optimal algorithm for learning quantum Hamiltonians in a physically motivated distance, that we introduce here: the \emph{time-constrained distance}. This distance captures the operational distinguishability of Hamiltonians through their induced dynamics when probed only up to a finite evolution time. We also introduce the \emph{temperature-constrained distance}, which quantifies differences between Hamiltonians through the thermal states they generate at finite inverse temperature.

\subsection{Problem statement}
Given a quantum system of $n$ qubits, its dynamics and thermal behaviour are governed by a Hamiltonian $H$, which formally is a self-adjoint matrix of $(\mathbb C^{2\times 2})^{\otimes n}$. As such, $H$ can be expanded in the Pauli basis, 
\begin{equation*}
    H=\sum_{P\in\{I,X,Y,Z\}^{\otimes n}}h_PP,
\end{equation*}
where $h_P=\Tr[HP]/2^n$ are the Pauli coefficients. The set $\mathcal H=\{P:h_P\neq 0\}$ is the Pauli support of the Hamiltonian, and the set $\mathcal{ H}_\eps=\{P:\ |h_P|\geq \eps\}$ is the $\eps$-effective support of the Hamiltonian.

By the Schr\"odinger equation, the dynamics of a system governed by a Hamiltonian $H$ is determined by the time-evolution operator $U(t)=e^{-iHt}$, meaning that if the state of the system at time $0$ is $\rho,$ then at time $t$ it will evolve to $U(t)\rho U^\dagger (t)$. Thus, a natural access model for Hamiltonians, introduced in \cite{huang2023heisenberg}, is to perform \emph{experiments} of the following kind: prepare a state $\rho,$ apply $U(t_1)$ (this involves making a query to $U(t_1)$, which in a lab can be implemented by letting the system evolve for time $t_1$), apply a unitary operator $V_1$ independent of $H$, query $U(t_2)$, apply a unitary operator $V_2$ independent of $H$, query $U(t_2)\dots$, and so on, then measure. See \cref{fig:experiment} for a pictorial representation of such an experiment. In this access model, there are different figures of merit that one may want to minimise. The two usually considered as the most important are the \emph{number of experiments}, and the \emph{total evolution time}, which is the sum of all times $t_i$ at which the algorithm queries $U(t_i)$. 

\begin{figure}[!ht]
    \centering
    \includegraphics[width=0.9\linewidth, page=1]{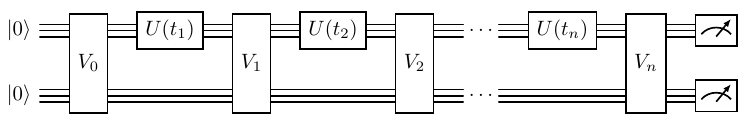}
    \caption{Scheme of an experiment to learn the dynamics $U(t)=e^{-iHt}$. $V_0$ prepares the initial state, $\rho$. Note that all unitaries $V_i$ do not depend on the underlying Hamiltonian $H$, which generates the evolutions $U(t_i)$.}
    \label{fig:experiment}
\end{figure}

Learning arbitrary $n$-qubit Hamiltonians in this model becomes prohibitive due to the exponential scaling of the complexity with the number of qubits \cite{bluhm2024hamiltonian}. For that reason, as done in recent work \cite{ma2024learning,zhao2025learning,arunachalam2025testing,hu2025ansatz,sinha2025improved}, we will restrict our attention to the case of $s$-sparse Hamiltonians, which are Hamiltonians with at most $s$ non-zero Pauli coefficients. We will also assume without loss of generality that $h_{I^{\otimes n}}=\Tr[H]/2^n=0,$ as two Hamiltonians that only differ on a multiple of identity produce the same time evolution and the same thermal states. Now, we are ready to formally state the Hamiltonian learning problem.\\

\fbox{\begin{minipage}{40em} \textbf{Hamiltonian Learning:}
Let $H$ be a traceless $n$-qubit $s$-sparse Hamiltonian. Let $\eps,\delta>0$. Let $d$ be the distance between Hamiltonians. The goal is to output another Hamiltonian $\widetilde{H}$ such that $\Pr\!\left[d(H,\widetilde{H})\leq \eps\right]\ge 1-\delta,$ by accessing $H$ via its time-evolution operator.
\end{minipage}
}

\subsection{Main result}
Motivated by the role of Hamiltonians in quantum system dynamics, we introduce the time-constrained distance (with budget time $T$) between two Hamiltonians as
\begin{equation*}
    d_T(H_1,H_2)=\max_{t\in [0,T]}\frac12\norm{\mathcal U_{H_1,t}-\mathcal U_{H_2,t}}_\diamond,
\end{equation*}
where $\mathcal U_{H,t}(\cdot)=e^{-itH}(\cdot)e^{itH}$ and $\norm{\cdot}_\diamond$ is the diamond norm. Due to the operational meaning of the diamond norm~\cite{aharonov1998,Sacchi2005}, the time-constrained distance has a clear physical interpretation: it quantifies the best probability of successfully discriminating the Hamiltonians by optimising over all protocols of the form ``prepare, evolve for a time $t\leq T$, and finally measure".

We also introduce an analogous temperature-constrained distance, motivated by the fact that the thermal equilibrium state of a system governed by $H$ at inverse temperature $\beta$ is the Gibbs state $\rho_H(\beta)=e^{-\beta H}/\Tr[e^{-\beta H}]$ (we refer to \cite{alhambra2023quantum} for background on quantum thermal equilibrium). Thus, we can naturally define the temperature-constrained distance (with inverse temperature budget of $B$) between two Hamiltonians as 
\begin{equation}
    d_B(H_1,H_2)=\max_{\beta\in [0,B]}\frac12\norm{\rho_{H_1}(\beta)-\rho_{H_2}(\beta)}_\trnorm.
\end{equation}
Due to the operational meaning of the trace norm~\cite{HOLEVO1973337,Helstrom}, the temperature-constrained distance has a clear physical interpretation: it quantifies the best probability of successfully discriminating the Hamiltonians by optimising over all protocols of the form ``prepare a thermal state at inverse temperature $\beta\leq B$, and measure".

Now, we are ready to state our main result (for a formal statement, see \cref{cor:learninphysicaldistances,cor:lb_experiments_timeconstrained,cor:lb_timeevol_timeconstrained} below). For simplicity, in the introduction, we will assume that the time budget and the inverse temperature budget are constant. However, in the main text, we explicitly control the dependencies on these parameters.

\begin{result}[Nearly-optimal Hamiltonian learning in physically  motivated distances]\label{result:main}
    Let $H$ be an $n$-qubit $s$-sparse Hamiltonian with bounded operator norm. Then, there is an algorithm that, with probability $\geq 0.9$, $\eps$-learns $H$ in the time-constrained distance, using $\widetilde O(s\log(1/\eps))$ experiments and $\widetilde O(s^2/\eps)$ total evolution time.
    
    Furthermore, the algorithm also $\eps$-learns in the temperature-constrained distance, is robust to a constant amount of error in state preparation and measurement, and uses no ancilla qubits.
    
    The number of experiments of any learning algorithm with $O(n)$ ancilla qubits must be at least $\Omega(s+(s/n)\cdot \log(1/\eps))$ and the total evolution time must be at least $\Omega(\sqrt s/\eps).$ 
\end{result}

To the best of our knowledge, this is the first nearly-optimal Hamiltonian learning result, with respect to the number of experiments, in a physically  meaningful distance. Notably, \cref{result:main} also achieves the Heisenberg scaling, i.e., the $1/\eps$ scaling of the total evolution time, that cannot be surpassed due to the Heisenberg principle. In addition, as we detail in \cref{sec:comparison}, the byproducts of Result~\ref{result:main} improve on several previous and concurrent works \cite{ma2024learning,zhao2025learning,arunachalam2025testing,hu2025ansatz,sinha2025improved}. 

\subsection{Techniques}
\subsubsection{Relating the physically meaningful distances to the operator norm}
We start by relating the time-constrained and temperature-constrained distance to the operator norm, which is easier to understand from a mathematical point of view. In particular, we prove the following (for a formal statement, see \cref{lem:upperboundtodT,lemma_2} and \ref{lemma_dB_upperbound}). 

\begin{result}[Relation of the operator norm and the  physically  motivated distances]\label{result:physicaldistancesvsopnorm}
    Let $H_1$ and $H_2$ be two $n$-qubit Hamiltonians with operator norm at most 1. Then,
    \begin{enumerate}
        \item For constant time budget $T>0$, $d_T(H_1,H_2)=\Theta(\norm{H_1-H_2}_{\opnorm})$. 
        \item For constant inverse temperature budget $B>0$, $d_B(H_1,H_2)=O(\norm{H_1-H_2}_{\opnorm})$. 
    \end{enumerate}
\end{result}
The upper bound $d_T(H_1,H_2)\leq T\norm{H_1-H_2}_{\opnorm}$ was already known (see for instance, \cite[Section 7.1.]{ma2024learning}). We further refine this result by proving the sharper estimate $d_T(H_{1},H_{2}) \le \sin \Big( \min \Big( \frac{\pi}{2}, T\Vert H_{1} - H_{2}\Vert_{\opnorm} \Big)\Big)$. .
Our contribution 
in establishing the first part of \cref{result:physicaldistancesvsopnorm} is to show that the time-constrained distance is not only upper bounded by the operator norm, but also equivalent, i.e., to show $\Omega(\norm{H_1-H_2}_{\opnorm})$. An upper bound of $d_B(H_1,H_2)\le e^{B}\norm{H_1-H_2}_{\opnorm}$ was known \cite[Lemma 16]{brandao2017quantum}, and here we exponentially improve the dependence on $B$ , specifically showing that $d_B(H_1,H_2)\le \frac{B}{2}\norm{H_1-H_2}_{\opnorm}$. 
We prove the latter by establishing the following upper bound on the trace distance between two Gibbs states, which can be of independent interest:
$$
    \left\|\frac{e^{H_1}}{\Tr e^{H_1}}-\frac{e^{H_2}}{\Tr e^{H_2}}\right\|_{\trnorm}
    \;\leq\;\|H_1-H_2\|_{\opnorm}.
$$

One may wonder if the temperature-constrained distance is also equivalent to the operator norm, however, we show that that is not the case (see \cref{sec:nolowboundfordB}).

In terms of learning, Result~\ref{result:physicaldistancesvsopnorm} means that learning in the time-constrained distance is equivalent to learning in the operator norm, and that upper bounds in the operator norm imply upper bounds in the temperature-constrained distance. Thus, from now on we focus on learning in the operator norm. 

\subsubsection{Upper bounds for learning}
We first note that proper learning in the $\ell_\infty$ norm of the Pauli coefficients, ($\norm{H}_{\ell_\infty}=\max_P|h_P|$), implies learning in the operator norm (with a different error parameter).\footnote{A learning algorithm for this problem is \emph{proper} if the output Hamiltonian is also $s$-sparse. It is \emph{improper} if the output Hamiltonian is allowed to not be $s$-sparse.} Indeed, a proper $(\eps/2s)$-learner in the $\ell_\infty$ norm outputs an $s$-sparse Hamiltonian $H'$ such that $|h_P-h_P'|\leq \eps/2s$ for all $P$ with high probability. Thus, by the triangle inequality and the fact that $\norm{P}_{\opnorm}\leq 1$, it follows that $\norm{H-H'}_{\opnorm}\leq \eps$ with high probability.   We also note that an improper learner in the $\ell_\infty$ norm can easily be made proper, by rounding the coefficients $h'_P\leq \eps/2s$ to $0.$ Thus, from now on, we focus on showing upper bounds for learning in the $\ell_\infty$ norm.

We divide our learning algorithm into two stages. In the first stage, we learn the effective Pauli support of the Hamiltonian. To be precise, we will find a set $\mathcal P\subseteq \{I,X,Y,Z\}^{\otimes n}$ satisfying, with high probability, that $\mathcal H_{\eps}\subseteq \mathcal P$ and $|\mathcal P|=\widetilde O(s).$\footnote{In other works, this stage is referred to as \emph{structure learning} \cite{bakshi2024structure,zhao2025learning}.} For this stage, we propose a novel \emph{isolation technique} that allows us to have query access to $e^{-ith_PP}$ for given $P\in\mathcal H.$ This isolation technique makes use of the sparsity of the Hamiltonian and it is inspired by a seminal result of Valiant and Vazirani, who showed that \emph{NP is as hard as detecting unique solutions} \cite{valiant1985np}. In the second stage, we devise a new subroutine to learn a single Pauli coefficient, and apply it to learn every element of $\mathcal P$.\footnote{A different algorithm for learning a single Pauli coefficient with the same number of experiments and the same total time evolution was already proposed in \cite{odake2024higher}. However, ours is simpler and it is also efficient in other figures of merit such as query complexity and classical post-processing time, and it is also robust against state preparation and measurement errors.} Before detailing these two stages of our algorithm, we state the result implied by it (see \cref{theo:learninginellinftynorm} for a formal statement).

\begin{result}[Upper bound for learning sparse Hamiltonians in $\ell_\infty$ norm]\label{result:upperinellinfty}
    Let $H$ be an $n$-qubit $s$-sparse Hamiltonian. Then, there is an algorithm that, with probability $\geq 0.9,$ $\eps$-learns every Pauli coefficient with $\widetilde O(s\log(1/\eps))$ experiments and $\widetilde O(s/\eps)$ total evolution time. 
\end{result}

\paragraph{Learning the effective Pauli support.}

The key novelty that allows us to learn the effective Pauli support of a $s$-sparse Hamiltonian $H$ is the isolation technique. This is a randomized protocol that for every $P\in\mathcal H$ prepares $e^{-ith_PP}$ with probability $\Omega(1/s)$. Once we have access to the isolated time evolution of $e^{-ith_PP}=\cos(h_Pt) I^{\otimes n}-i\sin(h_Pt)P$ for $P\in\mathcal H_\eps$, if we pick a uniformly random $t\in [\pi/4,1/\eps]$ and perform Pauli sampling of $e^{-ith_PP}$ (i.e., preparing the Choi-Jamiolkowski state of $e^{-ith_PP}$ and measure in the Bell basis, which outputs $P$ with probability $|\sin (h_Pt)|^2$), then we will obtain $P$ as the outcome with constant probability. Hence, repeating this process a constant number of times allows us to detect $P\in \mathcal H_\eps.$ Thus, if we repeat this iteration $\widetilde O(s)$ times and store the outcomes in a set $\mathcal P$, then we will have that $\mathcal H_{\eps}\subseteq \mathcal P$ with high probability. This way, it remains to explain how to access the isolated time evolution. To do that, we introduce some notation. For $Q_1,\dots, Q_r\in \{I,X,Y,Z\}^{\otimes n}$, we define $$H_{Q_1}=\frac{H+Q_1HQ_1}{2}\text{ and } H_{Q_1,\dots,Q_i}=\frac{H_{Q_1,\dots,Q_{i-1}}+Q_{i}H_{Q_1,\dots,Q_{i-1}}Q_i}{2}, \text{ for }2\leq i\leq r.$$

Now, we note three things. First, we can prepare $e^{-itH_{Q_1}}$ by making queries to the time evolution of $H$ with a total evolution time $t$. Indeed, this is true by Trotterization (see \cref{theo:trotterization}), that allows one to approximate $e^{-it(A+B)}$ by making queries to $e^{-itA}$ and $e^{-itB}$, and by the fact that $e^{-itQ_1HQ_1/2}=Q_1e^{-itH/2}Q_1$. The same is true for  $e^{-itH_{Q_1,\dots,Q_i}}$, so we may assume that we have access to $e^{-itH_{Q_1,\dots,Q_r}}$ without any overhead in the total evolution time and the number of experiments.

Second, notice that $H_Q=\sum_{P:[P,Q]=0}h_PP$, i.e., the action of $Q$ on $H$ \emph{kills} all the Paulis that do not commute with $P$. Similarly, $H_{Q_1,\dots,Q_r}=\sum_{P:[P,Q_1]=0\land \dots\land [P,Q_r]=0}h_PP$. 

Third, getting inspiration from \cite{valiant1984theory}, we prove that, given $P\in \mathcal H$, if we choose $r=\lceil\log(s)\rceil+2$ and pick $Q_1,\dots,Q_r$ uniformly at random from $\{I,X,Y,Z\}^{\otimes n}$, then we have that, with probability $\Omega(1/s)$, $P$ is the only element of $\mathcal H$ that commutes with all $Q_1,\dots,Q_r$ (see Lemma \ref{lem:survival} below). 

Putting these three observations together, it follows that we can prepare $e^{-ith_PP}$ with probability $\Omega(1/s).$ Now, we may state our result for learning the effective support of a $s$-sparse Hamiltonian (see \cref{theo:structurelearning} for a formal statement). 

\begin{result}[Learning the effective Pauli support of a sparse Hamiltonian]
    Let $H$ be an $n$-qubit $s$-sparse Hamiltonian with $|h_P|\leq 1$ for every $P$. Then, there is an algorithm making $\widetilde O(s)$ experiments and with $\widetilde O(s/\eps)$ total evolution time, that outputs a set $\mathcal P\subseteq \{I,X,Y,Z\}^{\otimes n}$ with size at most $\widetilde O(s)$ satisfying, with probability $\geq 0.9,$ that $\mathcal H_\eps\subseteq\mathcal P$. 
\end{result}

\paragraph{Learning a single Pauli coefficient.} Our strategy for learning a single Pauli coefficient is also based on the isolation technique. Now, given a known $P_0\in \{I,X,Y,Z\}^{\otimes n}$ and a $s$-sparse Hamiltonian, we pick uniformly random $Q_1,\dots, Q_r$ elements of $\{I,X,Y,Z\}^{\otimes n}$ satisfying $[Q_i,P_0]=0$, with $r=O(\log(s))$. Then, with high probability, we will have that for every $P\in (\mathcal H-\{P_0\})$ there is $i\in [r]$ such that $[P,Q_i]\neq 0$, so $H_{Q_1,\dots,Q_r}=h_{P_0}P_0$. As we can access $e^{-itH_{Q_1,\dots, Q_r}}$ via Trotterization as above, with high probability we have access to $e^{-ih_{P_0}P_0}$. Once there, we can determine $h_{P_0}$ efficiently. This way, we have the following result (see \cref{theo:singleparameterlearning} for a formal statement). 

\begin{result}[Learning a single Pauli coefficient of a sparse Hamiltonian]
    Let $H$ be an $n$-qubit $s$-sparse Hamiltonian with $|h_P|\leq 1$ for every $P$. Let $P_0\in \{I,X,Y,Z\}^{\otimes n}$. Then, there is an algorithm that performs $\widetilde O(\log(1/\eps))$ experiments and uses $\widetilde O(1/\eps)$ total evolution time, that outputs $h_{P_0}'\in\mathbb R$ satisfying, with probability $\geq 0.9,$ that $|h_{P_0}-h_{P_0}'|\leq \eps$. 
\end{result}

\paragraph{Other comments about our upper bounds for learning.} For the sake of simplicity, we have deferred to the main text the analysis of figures of merit other than number of experiments and total evolution time. However, we have performed an exhaustive complexity analysis of our algorithms, which are efficient (polynomial dependence with respect to all parameters) in other figures of merit such as classical post-processing or query complexity. Additionally, our algorithms are proper learners. Moreover, when they succeed, they output a Hamiltonian whose Pauli support is contained in the Pauli support of the Hamiltonian to be learned.

\subsubsection{Lower bounds for learning}
\paragraph{Lower bounds for number of experiments.} To show the lower bound for total number of experiments, we use two counting arguments, that together yield the following result (see \cref{thm:lb_experiments_1,thm:lb_experiments_2} for formal statements). 
\begin{result}[Lower bound on the number of experiments]\label{result:lbexperiments}
    An algorithm with $O(n)$ ancilla qubits that, with probability $\geq 0.9,$ $\eps$-learns in operator norm $s$-sparse Hamiltonians with $\norm{H}_{\opnorm}\leq 1$, must make $\Omega((s/n)\log(1/\eps))$ experiments. If $\eps\leq 1/s$ and $s=O(2^n)$, the number of experiments must be $\Omega(s).$
\end{result}
To sketch the proof of \cref{result:lbexperiments}, we start by noting that for each experiment with $O(n)$ ancilla qubits, one can extract $O(n)$ bits of information. Then, it suffices to find large sets of sparse Hamiltonians such that every pair of Hamiltonians in the set is far away from each other. 

If $\eps<1/s$ and $s=O(2^n)$, such a set can be taken to be $\{\sum_{P\in\mathcal H}\eps P:\, |\mathcal H|=s\}$. The size of this set is of the order of $4^{ns},$ so $\Omega(ns)$ bits are required to encode the elements of this set. In particular, the learning algorithm must extract $\Omega(ns)$ bits of information from the experiments, so the number of experiments must be $\Omega(s).$

We fix a subset $\mathcal H\subseteq \{I,X,Y,Z\}^{\otimes n}-\{I^{\otimes n}\}$, and we consider a set of maximum size of Hamiltonians supported on $\mathcal H$ such that every pair of the set is $\eps$-far from each other. 
A volume argument shows that the size of such a set is at least $\Omega((1/\eps)^{s}),$ so $\Omega(s\log(1/\eps))$ bits are required to describe it. This implies the lower bound of $\Omega((s/n)\cdot \log(1/\eps))$ experiments.

\paragraph{Time-evolution lower bounds.} For the time-evolution lower bounds we start by showing that for every algorithm that learns $s$-sparse Hamiltonians with error $\eps$ in $\ell_1$-norm of the Pauli coefficients ($\norm{H}_{\ell_1}=\sum_P|h_P|$), there is $P_0\in\{I,X,Y,Z\}^{\otimes n}$ such that the algorithm learns the Pauli coefficient of $P_0$ with error $\eps/s.$ By the hardness of learning a single Pauli coefficient (see \cite{huang2023heisenberg}, for instance), this implies a $\Omega(s/\eps)$ lower bound in the total time evolution, yielding the following result (see \cref{thm:lb_timeevol_c1} for a formal statement).

\begin{result}[Lower bound on the total time evolution]\label{result:lbl1normtimeevolution}
    An algorithm with $O(n)$ ancilla qubits that, with probability $\geq 0.9,$ $\eps$-learns in $\ell_1$-norm $s$-sparse Hamiltonians with $\norm{H}_{\ell_1}\leq 1$, must use $\Omega(s/\eps)$ total evolution time.
\end{result}

As for $s$-sparse Hamiltonians $\norm{H}_{\ell_1}/\sqrt{s}\leq \norm{H}_{\opnorm}\leq \norm{H}_{\ell_1}$ (see~\cref{lem:comp_l1_opnorm}), from \cref{result:lbl1normtimeevolution} it is easy to derive a $\Omega(\sqrt{s}/\eps)$ lower bound for learning in the operator norm. 

\subsection{Relation to previous work}\label{sec:comparison}
Since the seminal work of  Anshu, Arunachalam, Kuwahara, and Soleimanifar \cite{anshu2021sample}, there have been numerous papers about quantum Hamiltonian learning. Many of them assume prior knowledge of the Pauli support of the Hamiltonian, that the Hamiltonian is local (meaning that every Pauli in the support acts on a bounded number of qubits or smoother versions of this condition) and some geometrical structure (such as every qubit being acted on by a constant number of Pauli operators) \cite{anshu2021sample,Dutkiewicz.2023,huang2023heisenberg}. The first algorithms for learning without prior knowledge of the support of the Hamiltonian were proposed recently \cite{arunachalam2025testing,bakshi2024structure,ma2024learning}, and the locality assumption was also dropped in a recent series of work that only assumes sparsity \cite{arunachalam2025testing,zhao2025learning,hu2025ansatz,sinha2025improved}. 

Most of the aforementioned works learn in the $\ell_\infty$-norm or the $\ell_2$-norm of the Pauli coefficients. The $\ell_\infty$-norm lacks a direct physical interpretation, but is a coarse distance, easy to work with mathematically. The $\ell_2$-norm is less coarse and has an average-case physical interpretation: for a short time, two Hamiltonians that are close in $\ell_2$-norm produce time evolutions that are close on Haar-random inputs (see \cite[Section 7.2]{ma2024learning}). Only \cite{ma2024learning} learns in the $\ell_1$-norm of the Pauli coefficients, which has a worst-case physical interpretation due to \cref{result:physicaldistancesvsopnorm} and the fact that $\norm{\cdot}_{\opnorm}\leq \norm{\cdot}_{\ell_1}$.

To the best of our knowledge, the only prior optimal result for learning Hamiltonians is the one by Huang, Tong, Fang, and Su \cite{huang2023heisenberg}, who showed that for low interacting Hamiltonians (those that are local and where every qubit is acted on by a bounded number of Pauli operators) $\Theta(\eps^{-1})$ evolution time suffices to learn in the $\ell_\infty$-norm of the Pauli coefficients. 

Thus, \cref{result:main} gives the first learning algorithm for Hamiltonians that is optimal in one of the main figures of merit and considers a physically meaningful distance. Furthermore, it does not require locality and does not assume prior knowledge of the Pauli support. 

\paragraph{Direct comparison with prior work.} 
As most prior work considers distances such as the $\ell_\infty$ and $\ell_2$-norms of of the Pauli coefficients, we cannot directly compare \cref{result:main} with prior results. However, we can compare the intermediate results that led to \cref{result:main}. 

Regarding the lower bounds in the number of experiments, when writing the second version of this paper we discovered that a concurrent work by Chen, Li, and Slote gives a similar lower bound to our \cref{result:lbexperiments}, but in a related notion of complexity: the number of queries \cite{chen2025lower}. Their result and ours are incomparable, as they lower bound a smaller complexity measure than the number of experiments but they do not assume a bound in the quantum memory as we do (however, our assumption of $O(n)$ ancilla qubits is mild and satisfied by all known algorithms for this problem).\footnote{In addition, their lower bound works in a broader model than the one we consider and they show a lower bound for the easier task of learning a single Pauli coefficient.} Furthermore, Chen et al. note that if the minimum time at what the algorithm can query the time evolution is $\Omega(1/\poly(n))$, then one needs $\Omega(\poly(n))$ evolution time for learning $k$-local Hamiltonians (with $k$ large enough), solving a question posed by Bakshi, Liu, Moitra, Tang \cite{bakshi2024structure,FOCS2024Workshop}. Our \cref{result:lbexperiments} also implies the same conclusion, if we also assume $O(n)$ memory, giving an alternative and concurrent answer to the question of Bakshi, Liu, Moitra, Tang. 

For the lower bounds in the total time evolution, \cite[Theorem 11]{ma2024learning} showed a lower bound for learning $s$-sparse Hamiltonians with $\norm{H}_{\ell_\infty}\leq 1$ within the $\ell_1$-norm of the Pauli coefficients of $\Omega(s/\eps)$ total evolution time, in the case that there is a constant amount of state-preparation and measurement (SPAM) errors. Thus, \cref{result:lbl1normtimeevolution} improves \cite[Theorem 11]{ma2024learning} in two ways: $i)$ it substitutes the condition of $\norm{H}_{\ell_\infty}\leq 1$ by the more restrictive condition $\norm{H}_{\ell_1}\leq 1$; $ii)$ it drops the assumption of the presence of SPAM errors.

For the upper bounds, four other works considered the problem of learning $s$-sparse Hamiltonians in the $\ell_\infty$-norm of the Pauli coefficients \cite{zhao2025learning,arunachalam2025testing,hu2025ansatz,sinha2025improved}, so we can compare them with \cref{result:upperinellinfty}. As shown in \cref{table:comparison}, \cref{result:upperinellinfty} improves all of them. 

\begin{table}[!ht]
\centering
\begin{tabular}{c|c|c|}
\cline{2-3}
                               & \multicolumn{1}{c|}{Total evolution time} & \multicolumn{1}{c|}{Number of experiments}                       \\ \hline
\multicolumn{1}{|c|}{\cite{zhao2025learning}$^*$}     & $\widetilde O(s^3/\eps^4)$                           & $\widetilde O(s^4/\eps^4)$                                                  \\ \cline{1-1}
\multicolumn{1}{|c|}{\cite{arunachalam2025testing}$^\dagger$}      & $O(s^2/\eps^3)$                           & $O(s^4/\eps^4)$ \\ \cline{1-1}
\multicolumn{1}{|c|}{\cite{hu2025ansatz}}  & $\widetilde O(s^2/\eps)$                   & $\widetilde O(s^2\log(1/\eps)) $  
\\ \cline{1-1}
\multicolumn{1}{|c|}{\cite{sinha2025improved}$^\circ$}  & $\widetilde O(s/\eps)$                   & $\widetilde O(s^2\log(1/\eps)) $  
\\ \cline{1-1}
\multicolumn{1}{|c|}{\cref{result:upperinellinfty}} & $\widetilde O(s/\eps) $                   & $\widetilde O(s\log(1/\eps))$                                               \\ \hline
\end{tabular}
\caption{Comparison with prior upper bounds for learning $s$-sparse Hamiltonians with $\ell_\infty$-bounded Pauli coefficients in the distance induced by the $\ell_\infty$-norm of the Pauli coefficients.
\newline
\footnotesize \textsuperscript{*} The complexity of this result is originally stated in terms of $\norm{H}_{\opnorm}$, and here we used $\norm{H}_{\opnorm}\leq s\norm{H}_{\ell_\infty}$. It can be improved, paying huge constant factors, to $O((s/\eps^2)^{1+o(1)})$ total evolution time and $O((s/\eps)^{2+o(1)})$ experiments.\newline
\textsuperscript{$\dagger$} The complexity of this result is stated in terms of $\norm{H}_{\opnorm}$, and here we used $\norm{H}_{\opnorm}\leq s\norm{H}_{\ell_\infty}\leq s.$ This result makes use of a weaker access model, where only one query is made per experiment.
\newline
\textsuperscript{$^\circ$} This independent and concurrent result appeared three days before our work was posted on arXiv.}
\label{table:comparison}
\end{table}

\subsection{Future work.} We finish the introduction by gathering several questions motivated by our work. 

\begin{enumerate}
    \item {\bf Optimality of total evolution time.} Our main result is to find the optimal number of experiments to learn sparse Hamiltonians in physically motivated distances. The next natural question is to find the optimal evolution time for this task. 

    \item \textbf{Learning in physically motivated distances from Gibbs states.} In this work, we only consider the case of having access to the Hamiltonian via its time evolution, but, it would also be interesting to devise new algorithms for optimally learning the Hamiltonian in the distances that we introduce from copies of the Gibbs state. Previous works in this access model also focused on learning in $\ell_\infty$ and $\ell_2$ distance \cite{anshu2021sample, haah2022optimal, rouze2023learning, onorati2023efficient, bakshi2023learning, Gu2022Practical}, so they may yield suboptimal results when converted to the time- and temperature-constrained distances. 
    
    \item \textbf{A refined time-constrained distance.} The operational interpretation of our time-constrained distance comes from the scenario where one wants to distinguish two Hamiltonians by allowing to query their time evolution with time $t$ on a state and then measure. However, one may also define an analogue notion of time-constrained distance by allowing a more general scenario: prepare a state, evolve for time $t_1$, perform an operation independent of the Hamiltonian, evolve for $t_2,$ $\dots,$ and measure (here, $\sum_{i}t_i\leq t$). This also defines a natural notion of time constrained distance, for which our upper bound in terms of the operator norm distance still applies, but our proof of the lower bound does not. Thus, it would be interesting to show that lower bound, establishing the equivalence between the operator norm and this refined notion of time-constrained distance. 
\end{enumerate}

\paragraph{Acknowledgements.} We thank Lauritz Streck for useful conversations. 
F.E.G. was supported by the European union’s Horizon 2020 research and innovation programme under the Marie Sk{\l}odowska-Curie grant agreement no. 945045, and by the NWO Gravitation project NETWORKS under grant no. 024.002.003.
D.~G. was supported by NWO grant NGF.1623.23.025 (“Qudits in theory and experiment”) and NWO Vidi grant (Project No. VI.Vidi.192.109).
F.~A.~M. acknowledges financial support from the European Union (ERC StG ETQO, Grant Agreement no.\ 101165230). Views and opinions expressed are however those of the authors only and do not necessarily reflect those of the European Union or the European Research Council. Neither the European Union nor the granting authority can be held responsible for them.  F.~A.~M.~acknowledges financial support from the project: PRIN 2022 "Recovering Information in Sloppy QUantum modEls (RISQUE)", code 2022T25TR3, CUP E53D23002400006. N.~C.~and F.~A.~M.~thank University of Amsterdam and QuSoft for hospitality, as well as Vittorio Giovannetti for useful conversations. P.S. is supported by the Mike and Ophelia Lazardis Fellowship, the NSERC Alliance Grant and NSERC Discovery Grant.

\section{Preliminaries}\label{sec:preliminaries}

\subsection{Notation}
Given two bit strings $a,b\in\{0,1\}^n$, we define their dot product as $a\cdot b=\sum_{i\in [n]}a_i b_i\ (\mathrm{mod}\ 2)$. Given two bit strings $x=(a,b),\, y=(c,d)\in\{0,1\}^{2n}$, we define their symplectic product as $[x,y]=a\cdot c+ b\cdot d\ (\mathrm{mod}\ 2)$. We use $\{e_i\}_{i\in [n]}$ to refer to the $n$-dimensional canonical basis. We use $I,X,Y,Z$ to refer to the $2 \cross 2$ Pauli matrices. We use $\Id$ to refer to the identity matrix of a dimension that will be clear from the context. Given a matrix $A$ of $(\mathbb C^{2\times 2})^{\otimes n}$, we consider its expansion in the Pauli basis $$A=\sum_{P\in\{I,X,Y,Z\}^{\otimes n}}a_P P,$$ where $a_P=\tr[PA]/2^n$ are its Pauli coefficients. Given such a matrix, we use $\norm{A}_{\opnorm}$ to denote its operator norm when regarded as linear map from $\ell_2$ to $\ell_2,$ which coincides with its largest singular value. We use $\norm{A}_{\ell_p}=(\sum_P |a_P|^p)^{1/p}$ to refer to the $\ell_p$-norm of its Pauli coefficients. 

We will identify $n$-qubit Pauli strings with 2$n$-bit strings. For every Pauli string $P\in \{I,X,Y,Z\}^{\otimes n}$, there are a unique bit string $(a,b)\in \{0,1\}^{2n}$ and a unique number $\alpha_P\in\mathbb C$  such that 
$$P= \alpha_{P} X^{a_1}Z^{b_1}\otimes \dots \otimes X^{a_n}Z^{b_n}.$$
Note that these bit strings capture the commutation relations if the Pauli strings. Namely, if $P= \alpha_{P} X^{a_1}Z^{b_1}\otimes \dots \otimes X^{a_n}Z^{b_n}$ and $P'= \alpha_{P'} X^{a'_1}Z^{b'_1}\otimes \dots \otimes X^{a'_n}Z^{b'_n}$, then $P$ and $P'$ commute if $[(a,b),(a',b')]=0$, and anticommute otherwise.

\subsection{Hamiltonians}

An $n$-qubit Hamiltonian $H$ is a self-adjoint matrix of dimension $(\mathbb C^{2\times 2})^{\otimes n}$.  Its support is given by $\mathcal H\coloneqq\{P:|h_P|>0\}$ and its $\eps$-effective support is given by $\mathcal H_{\eps}\coloneqq\{P:|h_P|\geq \eps\}$. We assume that $|h_P|\leq 1$ for every $P\in \{I,X,Y,Z\}^{\otimes n}$. We note that a matrix $A$ is a Hamiltonian if and only if its Pauli coefficients are real numbers.

We will use the Taylor expansion of the exponential, that ensures that 
\begin{equation}\label{eq:Taylor}
    e^{-itH}=\Id-itH+R_1(t),
\end{equation}
for $t\leq 1/\norm{H}_{\opnorm},$  where $R_1(t)$ is the first-order remainder and satisfies $\norm{R_1(t)}_{\opnorm}\leq Ct^2\norm{H}_{\opnorm}^2$ for some constant $C\geq 1.$
\subsection{Access model} 
Hamiltonians govern the dynamics of quantum systems due to the Schr\"odinger equation. In particular, if a quantum system governed by a time-independent Hamiltonian $H$ and the state describing the system at time $0$ is $\rho,$ 
at time $t$ the state will have evolved to $U(t)\rho U^\dagger(t)$, where $U(t)=\exp(-itH)$ is the time evolution operator of $H$ at time $t$. 

Thus, a natural access model for Hamiltonians is to perform \emph{experiments} of the following kind: prepare a state $\rho,$ apply $U(t_1)$ (this is making a query to $U(t_1)$, which in a lab can be implemented by letting the system evolve for time $t_1$), apply a unitary operator $V_1$ independent of $H$, query $U(t_2)$, apply a unitary operator $V_2$ independent of $H$, query $U(t_2)$,$\dots$ and finally measure. 

In this access model, there are different figures of merit that one may care about. The two usually considered as the most important are the \emph{total evolution time}, which is the sum of all times at which the algorithm queries $U(t),$ and the \emph{number of experiments}. Other figures of merit that we will also keep track of are the \emph{number of queries,} the \emph{time resolution} (the minimum time at which the algorithm queries the time evolution operator), the \emph{classical post-processing time}, and the number of \emph{ancilla qubits}.

Finally, our algorithms will also be robust to \emph{state-preparation and measurement (SPAM) error}. Following \cite[Definition 4]{ma2024learning}, an experiment suffers from an $\eps$-amount of SPAM error if the channels applied before the first query to prepare the initial state and the channels applied after the last query to measure have in total $\eps$ error in diamond norm. We will say that an algorithm is \emph{robust} to an $\eps$ amount of SPAM error (or any other error) if the performance guarantees of the algorithm do not change in the presence of that error, maybe after increasing the complexities by constant factors. 

\paragraph{Trotterization.} 
Given access to $e^{-itA}$ and $e^{-itB}$ for two Hamiltonians $A$ and $B$ and arbitrary times $t$, Trotterization allows us to implement $e^{-it(A+B)}$ up to arbitrary error while also preserving the total time evolution and without using extra qubits. Thus, to analyze the number of experiments and the total time evolution required by our algorithms, if we have access to $e^{-itA}$ and $e^{-itB}$, we may assume access to $e^{-it(A+B)}$. However, the number of queries and the time resolution change. To be more precise, we will use the following result.

\begin{theorem}\cite[Corollary 2]{childs2021theory}\label{theo:trotterization}
    Let $t>0$, let $\eps>0$, let $H_1,\dots,H_R$ Hamiltonians acting on $n$-qubits and let $c=\max\{\norm{H_1}_{\opnorm},\dots,\norm{H_R}_{\opnorm}\}$. Let $l=\left\lceil O\left(\sqrt{(Rct)^3/\eps}\right)\right\rceil$ and $V=(e^{-itH_R/2l}\cdots  e^{-itH_1/2l}e^{-itH_1/2l}\cdots e^{-itH_R/2l})^l$. Then, 
    \begin{equation*}
        \norm{e^{-it(H_1+\dots+H_R)}\cdot e^{it(H_1+\dots+H_R)}-V\cdot V^{\dagger}}_{\diamond}\leq \eps.
    \end{equation*}
\end{theorem}
In other words, Theorem~\ref{theo:trotterization} shows that the time-$t$ evolution generated by the sum of $R$ Hamiltonians can be approximated to arbitrary accuracy using only evolutions under the individual Hamiltonians, with a total time evolution of $Rt$.
\subsection{Useful facts and lemmas}
\subsubsection{Concentration inequalities}
\begin{lemma}\label{lem:Chebychev}
    Let $X$ be a random variable with expectation $\mu$ and variance $\sigma^2$. Let $k>0.$ Then, 
    $$\Pr[|X-\mu|\geq k\sigma]\leq \frac{1}{k^2}.$$
\end{lemma}

\begin{lemma}[Hoeffding bound]\label{lem:hoeffding}
Let~$X_1,\dots,X_m$ be independent-random variables that satisfy $-a_i\leq |X_i|\leq a_i$ for some $a_i>0$.
Then, for any $\tau > 0$, we have
$$
\Pr\Big[\Big|\sum_{i\in [m]} X_i-\sum_{i\in [m]}\mathbb E[X_i]\Big| > \tau\big]
\leq
2\exp\left(-\frac{\tau^2}{2(a_1^2 + \cdots + a_m^2)}\right).
$$
\end{lemma}

\subsubsection{Valiant-Vazirani lemma}

Here, we state the lemma by Valiant and Vazirani to show that NP is as easy as detecting unique solutions. We will use this lemma to learn the effective support of a sparse Hamiltonian  \cite[Theorem 2.3]{valiant1985np}. 

\begin{lemma}\label{lem:ValiantVazirani}
    Let $\mathcal X\subseteq \{0,1\}^n-\{0^n\}$. Let $y_1,\dots,y_r$ be iid uniformly random elements of $\{0,1\}^n.$ Let $\mathcal S=\{x\in\mathcal X:\, x\cdot y_1=\dots=x\cdot y_r=0\}$. Then, 
    $$\mathbb E[|\mathcal S|]=2^{-r}|\mathcal X|\quad \text{and}\quad \mathrm{Var}[|\mathcal S|]=2^{-r}(1-2^{-r})|\mathcal X|.$$
\end{lemma}

\begin{corollary}\label{cor:ValiantVazirani}
    Let $\mathcal X\subseteq \{0,1\}^n-\{0^n\}$. Let $y_1,\dots,y_{\lceil\log(|\mathcal X|)\rceil+2}$ be iid uniformly random elements of $\{0,1\}^n.$ Let $\mathcal S=\{x\in\mathcal X:\, x\cdot y_1=\dots=x\cdot y_{\lceil\log(|\mathcal X|)\rceil+2}=0\}$. Then, 
    $$\Pr[|\mathcal S|=0]\geq \frac{1}{2}.$$
\end{corollary}
\begin{proof}
    By Lemma \ref{lem:ValiantVazirani}, we have that 
    $$\mathbb E[|\mathcal S|]\le\frac{1}{4}\quad \text{and}\quad \mathrm{Var}[|\mathcal S|]< \frac{1}{4}.$$
    Then, by Chebychev's inequality (Lemma~\ref{lem:Chebychev})
    $$\Pr\bigg[|\mathcal{S}|\geq \underbrace{\frac{1}{4}+\frac{\sqrt{2}}{2}}_{\ge\mu+\sqrt{2}\sigma}\bigg]\leq \frac{1}{2}.$$
    As $1/4+\sqrt{2}/2<1$ and $|\mathcal S|$ only takes integer values, we have that $\Pr[|\mathcal S|=0]\geq 1/2.$
\end{proof}

\begin{corollary}\label{cor:ValiantVaziranidelta}
    Let $\mathcal X\subseteq \{0,1\}^n-\{0^n\}$. Let $\delta<1$. Let $y_1,\dots,y_{\lceil\log(|\mathcal X|)+\log(1/\delta)+2\rceil}$ be iid uniformly random elements of $\{0,1\}^n.$ Let $\mathcal S=\{x\in\mathcal X:\, x\cdot y_1=\dots=x\cdot y_{\lceil\log(|\mathcal X|)+\log(1/\delta)+2\rceil}=0\}$. Then, 
    $$\Pr[|\mathcal S|=0]\geq 1-\delta.$$
\end{corollary}
\begin{proof}
    By Lemma \ref{lem:ValiantVazirani}, we have that $$\mathbb E[|\mathcal S|]\le \delta/4\quad \text{and}\quad \mathrm{Var}[|\mathcal S|]< \delta/4.$$
    Then, by Chebychev's inequality (Lemma~\ref{lem:Chebychev}), 
    $$\Pr\bigg[|\mathcal{S}|\geq \underbrace{\frac{\delta}{4}+\frac{1}{2}}_{\ge\mu+(1/\sqrt{\delta})\sigma}\bigg]\leq \delta.$$
    As $\delta/4+1/2<1$ and $|\mathcal S|$ only takes integer values, we have that $\Pr[|\mathcal S|=0]\geq 1-\delta.$
\end{proof}

\subsubsection{Pauli sampling.} Suppose $U$ is a unitary and we write out its Pauli decomposition as $U=\sum_P u_PP$, then by Parseval's identity  $\sum|u_{P}|^2=\Tr[U^\dagger U]/2^n=1$, i.e.,  $\{|u_{P}|^2\}_P$ is a  \emph{probability distribution}. We will be using the fact below extensively.
 \begin{fact}
 \label{fact:bellsamplingU}
     Given access to a unitary $U$, one can sample from the distribution  $\{|u_P|^2\}_P$.
 \end{fact}
 \begin{proof}
     The proof simply follows by applying $U\otimes \Id_{2^n}$ to $n$ EPR pairs (i.e., preparing the Choi-Jamiolkowski state of $U$) and measuring in the Bell basis, because $$
 U\otimes\Id_{2^n}\ket{\mathrm{EPR}_n}=\sum_{P\in\{I,X,Y,Z\}^{\otimes n}}u_P \mathop{\bigotimes}_{i\in [n]}(
 P_i\otimes I\ket{\mathrm{EPR}}),
 $$
and the Bell states can be written as $P\otimes I\ket{\mathrm{EPR}}$ for $P\in\{I,X,Y,Z\}$, where $\ket{\mathrm{EPR}}=(\ket{00}+\ket{11})/\sqrt 2$.
 \end{proof}

\subsubsection{The Clifford group} 
We recall that the $n$-qubit Clifford group is the subgroup of the $n$-qubit unitaries that, via conjugation, map $n$ qubit Pauli operators to $n$-qubit Pauli operators, up to a phase. It can alternatively be described as the group of unitary operators that are generated by $H,S$ and $CX$ gates acting upon any of the $n$ qubits. Note that tensor products of Clifford Unitaries are also Clifford Unitaries.
\begin{lemma}\label{lem:clifford_orbit}
    Let $P_1,\ P_2\in\{I,X,Y,Z\}^{\otimes n}-\{I^{\otimes n}\}$. Then, there is a Clifford unitary $C$ and a complex number $c$ such that $CP_1C^{\dagger}=cP_2$. Furthermore, up to a phase, $P\to C^{\dagger}P C$ is a bijection of $\{I,X,Y,Z\}^{\otimes n}$.
\end{lemma}
\begin{proof}
Any two Pauli operators $P_1$ and $P_2$ for which there exists Clifford unitary $C$ such that $CP_1C^{\dagger}=cP_2$ are said to be Clifford equivalent, and it can be verified that this is indeed an equivalence relation.

For the first part, it suffices to show that each non identity Pauli is Clifford equivalent to $X\otimes I^{\otimes (n-1)}$ up to a phase.
Now, note that all non identity single Qubit Pauli's are equivalent to each other, because 
        \[HXH\propto Z \text{ and }  SXS^\dagger\propto Y.\]
Thus, we have that any $n$-qubit Pauli operator is Clifford equivalent to a tensor product of $X$s and $I$s. We continue by noting that $CX(X\otimes X)CX= I \otimes X$, so every tensor product of $X$s and $I$s is Clifford equivalent to a tensor product of just one $X$ and $I$s.  Finally, by conjugating with swap gates (which also lies in the Clifford group), we can move that $X$ to the first qubit. This shows that all Pauli strings are Clifford equivalent to $X\otimes I^{\otimes (n-1)}$, as claimed.

For the second part, it suffices to note that, by definition, Clifford circuits map, by conjugation, Pauli strings to Pauli strings (up to phases) and that this conjugation operatoion is invertible.
\end{proof}

\subsubsection{Distances between quantum objects}
We recall here two well-known results which will be useful in the following (for proofs, see \cite[Proposition~1.6 and Lemma~3.2(d)]{haah2023query}) 
\begin{lemma}\label{eq_dia_inf}
    For any two unitaries $U$ and $V$, it holds that
    \begin{equation}
        \frac{1}{2}\left\| U(\cdot)U^\dagger - V(\cdot)V^\dagger \right\|_\diamond 
        \;\le\; \min_{\phi\in\mathbb{R}}\|Ue^{i\phi}-V\|_{\opnorm} 
        \;\le\; \left\| U(\cdot)U^\dagger - V(\cdot)V^\dagger \right\|_\diamond .
    \end{equation}
\end{lemma} 
\begin{lemma}\label{lemma_commm}
    Let $X,Y$ be anti-Hamiltonians. Then
    \begin{equation}
        \big\| e^X e^Y - e^{X+Y} \big\|_{\opnorm} 
        \;\le\; \tfrac{1}{2}\,\|[X,Y]\|_{\opnorm}.
    \end{equation}
\end{lemma}

We state, and prove for completeness, a few folklore results about relations between $\norm{H}_{\opnorm}$ and various $\ell_p$ norms of the Pauli coefficients.
\begin{lemma}\label{lem:comp_linf_opnorm}
    For an $s$-sparse Hamiltonian $H$,
    \[s\norm{h}_{\ell_{\infty}}\geq \norm{H}_{\opnorm}\geq \norm{h}_{\ell_\infty}.\]
\end{lemma}
\begin{proof}
    The first inequality follows by the triangle inequality. For the second, observe that for $H=\sum h_PP$. Then, for each $P$, we have

    \[h_P=\frac{1}{2^n}\Tr[HP]\leq \norm{HP}_{\opnorm}\leq \norm{H}_{\opnorm}\cdot \norm{P}_{\opnorm}=\norm{H}_\opnorm.\]
\end{proof}
\begin{lemma}\label{lem:comp_l1_opnorm}
    For an $s$-sparse Hamiltonian $H$,
    \[\norm{h}_{\ell_1}\geq \norm{H}_{\opnorm}\geq \norm{h}_{\ell_2}\geq \frac 1 {\sqrt s}\norm{h}_{\ell_1}.\]
    
\end{lemma}
\begin{proof}
    The first inequality follows by the triangle inequality and third by Cauchy-Schwarz inequality. For the second inequality, let $H=\sum h_P P$. Observe that

    \begin{align*}
        \norm{H}_{\opnorm}^2&=\norm {H^2}_{\opnorm}\geq \frac {1}{2^n}\Tr[H^2]=\sum_P h_P^2=\norm{H}_{\ell_2}^2,
    \end{align*}
    where we have used that the sum of eigenvalues of $H^2$ is at most $2^n$ times the operator norm of $H^2$, and that $\Tr[PQ]=\delta_{P,Q}2^n$ for every pair of Pauli operators $P,Q$.
\end{proof}

    \subsubsection{Nayak's bound.} 
    For our lower bounds, we will need the following lemma showed for the communication complexity of the input guessing game proved by Nayak \cite{nayak1999optimal}. In the input guessing game there are two players, Alice and Bob. Alice receives an uniformly random input from $[K]$ and the goal is that Bob guesses that input via a two-way communication protocol. 

    \begin{lemma}\label{lem:Nayaksbound}
        Any quantum algorithm that solves the input guessing game for $K\in\mathbb N$ with success probability at least 0.9  must exchange $\Omega(\log K)$ qubits.
    \end{lemma}

\section{Upper bound for learning}
\subsection{Learning the effective support}
In this section, we show that the effective support of an $s$-sparse Hamiltonian can be learned efficiently. The key idea in the proof of this result (\cref{theo:structurelearning} below) is a random procedure that, for every $P\in \mathcal H$ allows us to prepare the \emph{isolated Hamiltonian evolution} $e^{-ith_PP}$ with some non-negligible probability. Once we have access to  $e^{-ith_PP}$, it is easy to detect $P$ if $|h_P|\geq \eps.$ To sketch the procedure to prepare the isolated Hamiltonian evolution, we should introduce some notation. 

\begin{definition}
    Let $Q_1,\dots, Q_r\in \{I,X,Y,Z\}^{\otimes n}$ and let $H$ be an $n$-qubit Hamiltonian. We define $$H_{Q_1}=\frac{H+Q_1HQ_1}{2}\text{ and } H_{Q_1,\dots,Q_i}=\frac{H_{Q_1,\dots,Q_{i-1}}+Q_{i}H_{Q_1,\dots,Q_{i-1}}Q_i}{2}, \text{ for }2\leq i\leq r.$$
\end{definition}

First, note that we can prepare $e^{-itH_{Q_1}}$ by making queries to the time evolution of $H$ with a total evolution time $t$ and without using extra qubits. Indeed, this is true because of \cref{theo:trotterization} and the fact that $e^{-itQ_1HQ_1/2}=Q_1e^{-itH/2}Q_1$. The same is true for  $e^{-itH_{Q_1,\dots,Q_i}}$, so we may assume that we have access to $e^{-itH_{Q_1,\dots,Q_r}}$ without any overhead in the total evolution time nor in the number of extra qubits.

Second, notice that $H_Q=\sum_{P:[P,Q]=0}h_PP$, i.e., the action of $Q$ on $H$ \emph{kills} all the Paulis that do not commute with $P$. Similarly, $H_{Q_1,\dots,Q_r}=\sum_{P:[P,Q_1]=0\land \dots\land [P,Q_r]=0}h_PP$. 

Third, we prove (see Lemma \ref{lem:survival} below) that, given $P\in \mathcal H$, if we choose $r=\lceil\log(s)\rceil+2$ and pick $Q_1,\dots,Q_r$ uniformly at random from $\{I,X,Y,Z\}^{\otimes n}$, then we have that with probability $\Omega(1/s)$ $P$ is the only element of $\mathcal H$ that commutes with all $Q_1,\dots,Q_r.$ Putting these three observations together, it follows that we can prepare $e^{-ith_PP}$ with probability $\Omega(1/s).$

\begin{lemma}\label{lem:survival}
    Let $r=\lceil \log(s)\rceil+2$. Let $Q_1,\dots,Q_{r}$ iid uniformly distributed elements of $\{I,X,Y,Z\}^{\otimes n}$. Let $H$ be an $s$-sparse Hamiltonian and let $P\in \mathcal H$. Then, the probability that $P$ is the only element of $\mathcal H$ that commutes with all $Q_1,\dots,Q_r$ is $\Omega(1/s)$.
\end{lemma}

\begin{proof}
    For the second part, without loss of generality, we can assume $P=Z\otimes I^{\otimes {(n-1)}}$. Indeed, for any other non-identity $P\in \{I,X,Y,Z\}^{\otimes n}$, by Lemma \ref{lem:clifford_orbit}, there is a Clifford operation $C$ such that $CPC^{\dagger}=Z\otimes I^{\otimes (n-1)}$, so instead of working with $H=\sum a_iP_i$, we can assume that we are working with $CHC^{\dagger}=\sum a_i CP_iC^{\dagger}$. This is does not affect the commutation relations nor the randomness of $Q$, because $Q$ and $P$ commute if and only if $CPC^\dagger$ and $CQC^\dagger$ commute, and $Q$ is uniformly random if and only if $C^\dagger QC$ is uniformly random.
    
    Now, we associate every Pauli string with a $2n$-bit string as in \cref{sec:preliminaries}. We  then write $H$ as 
    \begin{equation*}
        H=\sum_{i\in [s]}h_i\alpha_i X^{a^i_1}Z^{b^i_1}\otimes \dots\otimes X^{a^i_n}Z^{b^i_n},
    \end{equation*}
    for some $x^1=(a^1,b^1),\dots,x^s=(a^s,b^s)\in \{0,1\}^{2n}-\{0^{2n}\}$ and $\alpha_1,\dots,\alpha_s\in\mathbb C.$ We may assume without loss of generality that $P=Z\otimes I^{\otimes n}$ corresponds to $x^1=e_{n+1}.$ Also, we write $Q_l$ as $$\beta_l X^{c^l_1}Z^{d^l_1}\otimes \dots\otimes X^{c^l_n}Z^{d^l_n},$$
    for some $y^1=(c^1,d^1),\dots,y^{r}=(c^{r},d^{r})\in \{0,1\}^{2n}$ and $\beta_1,\dots,\beta_{r}\in\mathbb C.$
    
    With this notation, we can write 
    \begin{align*}
         &\Pr[P\text{ is the only element of }\mathcal H\text{ that commutes with } Q_1,\dots,Q_{r}]\\
         &=\Pr[\{[e_{n+1},y^j]=0\ \forall j\in [r]\}\land \{\forall\  2\leq i\leq s,\ \exists\ j\ [x^i,y^j]=1\}]\\
        &=\Pr[\{y^j_{1}=0\ \forall j\in [r]\}\land \{\forall\  2\leq i\leq s,\ \exists\ j\ [x^i,y^j]=1\}].
    \end{align*}   
    Now, we write
    \begin{align*}
        \Pr[H_{Q_1,\dots,Q_{r}}=h_PP]&=\underbrace{\Pr[y^j_1=0\ \forall j\in [r]]}_{(*)}\cdot \underbrace{\Pr[\{\forall\  2\leq i\leq s,\ \exists\ j\ [x^i, y^j]=1\}|\{y^j_1=0\ \forall j\in [r]\}]}_{(**)}.
    \end{align*}
    To finish, it suffices to show that $(*)=\Omega(1/s)$ and that $(**)=\Omega(1).$ We prove both things separately. 

    First, to show that $(*)=\Omega(1/s)$, note that for a uniformly random $y\in \{0,1\}^{2n}$ satisfies $y_1=0$ with probability $1/2$, so
    \begin{equation*}
        (*)=\frac{1}{2^r}=\frac{1}{4s}.
    \end{equation*}

    Second, to show that $(**)=\Omega(1)$ we notice that 
    
    \begin{align*}
        (**)=\Pr[\forall\  2\leq i\leq s,\ \exists\ j\ \sum_{l\neq n+1}x^i_ly^j_l=1]
    \end{align*}
    where $y^j_l$ are iid uniformly random elements of $\{0,1\}.$    
    Finally, we have that $$(x^i_1,\dots,x^i_{n},x^i_{n+2},\dots,x^i_{2n})\neq 0^{2n-1}$$ because $x^i\neq 0^{2n}$ and $x^i\neq x^1=e_{n+1},$ so we can apply Corollary \ref{cor:ValiantVazirani} to conclude that $(**)\geq 1/2=\Omega(1)$. 
\end{proof}

To control the query overhead required to apply $e^{itH_{Q_1,\dots,Q_r}}$ we will use the following consequence of  \cref{theo:trotterization}. 

\begin{lemma}\label{lem:TrotterizationOfHQ}
    Let $Q_1,\dots,Q_r\in \{I,X,Y,Z\}^{\otimes n}$, let $H$ be an $n$-qubit Hamiltonian with $|h_P|\leq 1$ for every $P\in \{I,X,Y,Z\}^{\otimes n}$. Then, one can implement $e^{-itH_{Q_1,\dots,Q_{r}}}$ up to error $\eps$ in diamond norm with $T=O\left(2^r\sqrt{(\norm{H}_{\opnorm}t)^3/\eps}\right)$ queries to the time evolution operator of $H$ at time $\Omega\left((1/2^r)\cdot \sqrt{\eps/(t\norm{H}_{\opnorm}^3)}\right)$.
\end{lemma}
\begin{proof}
    For $S\subseteq [r]$, we order its elements $i_1,\dots,i_w$ in a increasing order, i.e., $i_w>i_{w-1}>\dots>i_1$. We define $H_S=Q_{i_w}\dots Q_{i_1}HQ_{i_1}\dots Q_{i_w}/2^r$. Now, note that 
    \begin{equation*}
        H_{Q_1,\dots,Q_r}=\sum_{S\subseteq [r]} H_S,
    \end{equation*}
    which can be easily verified by induction on $r$. Thus, applying \cref{theo:trotterization} with $R=2^r$ and $c=\norm{H}_{\opnorm}/R,$ ensures that $V=(e^{-tH_\emptyset/2l}\cdots  e^{-tH_{[r]}/2l}e^{-tH_{[r]}/2l}\cdots e^{-tH_\emptyset/2l})^l$, for $l=O\left(\sqrt{(\norm{H}_{\opnorm}t)^3/\eps}\right)$, is $\eps$-close in diamond distance to $e^{-itH_{Q_1,\dots,Q_{r}}}$. Finally, we conclude by noticing that every $e^{itH_S/2l}$ can be implemented with one query to the time evolution operator of $H$ at time $t/2^{r+1}l$ so $V$ can be implemented with $2^{r+1}l=2^rO\left(\sqrt{(\norm{H}_{\opnorm}t)^3/\eps}\right)$ queries to the time evolution operator of $H.$ 
\end{proof}

Now, we are ready to show that the effective support of a sparse Hamiltonian can be learned efficiently. 
\begin{theorem}\label{theo:structurelearning}
    Let $H$ be a $s$-sparse $n$-qubit Hamiltonian and let $\eps>0$. Then, the \cref{alg:theorem34} outputs a set $\mathcal P\subseteq \{I,X,Y,Z\}^{\otimes n}$ such that $\mathcal H_{\eps}\subseteq \mathcal{ P}$ and $|\mathcal P|=O(s\log(s/\delta))$ with $O(s\log(s/\delta))$ experiments and $O(s\log(s/\delta)/\eps)$ total evolution time with probability at least $1-\delta$. 
    
    Furthermore, the \cref{alg:theorem34} uses $n$ extra qubits, it only makes $O(s^2(\norm{H}_{\opnorm}/\eps)^{3/2}\log(s/\delta))$ queries, the time resolution is $\Omega(\eps^{1/2}/(s\norm{H}_{\opnorm}^{3/2}))$, it only uses $O(ns\log(s/\delta))$ classical post-processing time, it is non-adaptive, and is robust against a constant amount of SPAM errors. 
\end{theorem}

\begin{algorithm}[H]
\caption{Learning the Pauli support of the Hamiltonian}
\label{alg:theorem34}
\begin{algorithmic}[1]
\Require $s \in \mathbb{N}$, $n \in \mathbb{N}$, $H$ an $s$-sparse $n$-qubit Hamiltonian, 
error parameter $\varepsilon > 0$, failure probability $\delta \in (0,1)$
\Ensure A set $\mathcal{P} \subseteq \{I, X, Y, Z\}^{\otimes n}$ of size $O(s \log(s/\delta))$ such that 
$\mathcal{H}_\varepsilon \subseteq \mathcal{P}$ with probability at least $1-\delta$
\State Initialize $\mathcal{P} \gets \varnothing$.
\For{$i=1$ to $T = \Theta(s \log(s/\delta))$}
  \State Pick $Q_1, \ldots, Q_{\log(s)+2}$ i.i.d.\ elements of $\{I, X, Y, Z\}^{\otimes n}$.
  \State Pick uniformly random $t \in [\pi/4, 1/\varepsilon]$.
  \State Sample $P$ be performing Pauli sampling of $e^{-it H_{Q_1,\ldots,Q_{\log(s)+2}}}$ (see \cref{fact:bellsamplingU}).
  \State $\mathcal{P} \gets P$.
\EndFor \\
\Return $\mathcal{P}$
\end{algorithmic}
\end{algorithm}
\begin{proof}[Proof of \cref{theo:structurelearning}]  
    We show correctness of \cref{alg:theorem34} and perform a complexity analysis.
    
    \textbf{Correctness.} Now, we prove the correctness, i.e., that $\mathcal P$ contains $\mathcal H_{\eps}$ with probability $\geq 1-\delta$. Fix $P\in\mathcal H_{\eps}$. Using union bound over the at most $s$ elements of $\mathcal H_{\eps}$, it suffices to prove that $P\in \mathcal P$ with probability at least $1-\delta/s$. By Lemma \ref{lem:survival}, for every choice of $Q_1,\dots,Q_{\log(s)+2}$ the probability of $H_{[Q_1,\dots,Q_{\log(s)+2}]}= h_PP$ is at least $\Omega(1/s)$. Assume this happens in some iteration. Then,
    \begin{equation*}
        e^{-itH_{Q_1,\dots, Q_{\log(s)+2}}}=\cos(h_pt)I+i\sin(h_Pt) P.
    \end{equation*}
    Hence, the probability of observing $P$,  when measuring the Choi state, is $|\sin(h_Pt)|^2$. As $|h_P|\in [\eps,1]$, we have that with constant probability over the choice of $t\in [\pi/4,1/\eps]$ it is satisfied that $|(\pi/2-h_Pt)\mod 2\pi|\leq \pi/4.$ Thus, in each iteration, there is an $\Omega(1/s)$ probability of sampling $P$. Therefore, the probability that we do not sample $P$ in any of the iterations is
    \begin{equation*}
        \left(1-\Omega\left(\frac{1}{s}\right)\right)^{\Theta(s\log(s/\delta))}\leq O\left(\frac \delta s\right)\,.
    \end{equation*}
    
    \textbf{Complexity analysis.} Finally, we analyze the complexity. Note that the above process just uses $O(s\log(s/\delta))$ experiments, and $O(s\log(s/\delta)\cdot 1/\eps)= O(s\log(s/\delta)/\eps)$ total evolution time. So far, we have assumed perfect access to $e^{-itH_{Q_1.\dots.Q_r}}$ for $t\in [\pi/4,1/\eps]$, which does not affect the number of experiments nor the total evolution time, but it does affect the number of queries and time resolution. For the algorithm to be correct, we just need to ensure that in the case that $H_{Q_1,\dots,Q_r}=h_PP$ and $|\sin(h_Pt)|^2\geq 1/2$, the probability of measuring $P$ is greater than a constant, say $1/4.$ Without error, we know that if all steps succeed, that probability is $|\sin(h_Pt)|^2\geq 1/2$, so it suffices to approximate  $e^{-itH_{Q_1.\dots.Q_r}}$ up to error $1/8$ and we can allow for a SPAM error of $1/8$. The approximation of $e^{-itH_{Q_1.\dots.Q_r}}$ can be done with a query overhead of $O(2^r(\norm{H}_{\opnorm}/\eps)^{3/2})=O(s(\norm{H}_{\opnorm}/\eps)^{3/2})$ queries, thanks to Lemma \ref{lem:TrotterizationOfHQ}. As the time for which we have to query $H_{Q_1,\dots,Q_r}$ is $O(1/\eps)$, by Lemma \ref{lem:TrotterizationOfHQ} we also have that the time resolution $\Omega(\eps^{1/2}/(s\norm{H}_{\opnorm}^{3/2}))$. As we only need one query to $e^{-itH_{Q_1.\dots.Q_r}}$ per experiment, the total number of  queries is $ O(s(\norm{H}_{\opnorm}/\eps)^{3/2})O(s\log^3(s/\delta))=O(s^2(\norm{H}_{\opnorm}/\eps)^{3/2}\log(s/\delta))$. Regarding the classical post-processing times, it is $O(ns\log(s/\delta))$. This is true because the only classical post-processing done after each of the $O(s\log(s/\delta))$ experiments is storing the sampled Paulis, which can be represented with $n$ bits. 
\end{proof}

\subsection{Single parameter learning}\label{sec:singleparameterlearning}
In this section, we show how to learn a single Pauli coefficient of a sparse Hamiltonian (see \cref{theo:singleparameterlearning}). Before that, we prove an auxiliary lemma that allows to estimate a single parameter of an isolated Hamiltonian evolution $e^{-ith_PP}$, provided that $h_P$ is small enough. 

\begin{lemma}\label{lem:singleparameterlearningconstrained}
    Let $P\in \{I,X,Y,Z\}^{\otimes n}$, let $\eps,\delta>0$, let $H=h_PP$ for some $|h_P|\leq 10 \eps$. Then, there is an algorithm that with $O((1/\eps)\cdot\log(1/\delta))$ total evolution time and $O(\log(1/\delta))$ experiments outputs $h_P'$ such that $|h_P-h_P'|\leq \eps$ with probability $\geq 1-\delta.$

    The algorithm uses just $n$ extra qubits, it only makes $O(\log(1/\delta))$ queries to the time evolution operator of $H,$ the time resolution is $\Omega(1/\eps)$, it only uses $O(\log(1/\delta))$ classical post-processing time, and is robust against a constant amount of SPAM error.
\end{lemma}

\begin{algorithm}[H]
\caption{Learning a small isolated single Pauli coefficient $h_P$}
\label{alg:singleparam_small}
\begin{algorithmic}[1]
\Require Pauli operator $P \in \{I, X, Y, Z\}^{\otimes n}$, error $\varepsilon > 0$, failure probability $\delta \in (0,1)$, Hamiltonian $H=h_PP$ with $|h_P| \leq 10\varepsilon$
\Ensure Estimate $h_P'$ such that $|h_P - h_P'| \leq \varepsilon$ with probability $\ge 1-\delta$
\State Set $t \gets 1/(800 \varepsilon C)$, where $C$ is the contant in \cref{eq:Taylor} 
\Comment{First stage: estimate $|h_P|$}
\State Perform Pauli sampling of $U = e^{-itH}$ a total of $O(\log(1/\delta))$ times.
\State Get empirical estimate $|u_P|'$ of $P$-th Pauli coefficient of $U$.
\State Set $|h_P|' \gets |u_P|'/t$.
\Comment{Second stage: determine sign of $h_P$}
\State Perform Pauli sampling of $\widetilde{U} = e^{-it(H + |h_P|'P)}$ a total of $O(\log(1/\delta))$ times (using Trotterization). 
\State Get empirical estimate $|\widetilde{u}_P|'$ of $P$-th Pauli coefficient of $\widetilde{U}$. 
\State Set $|\widetilde{h}_P|' \gets |\widetilde{u}_P|'/t$.
\If{$|\widetilde{h}_P|' \geq \varepsilon/2$}
    \State \Return $h_P' \gets |h_P|'$
\Else
    \State \Return $h_P' \gets -|h_P|'$
\EndIf
\end{algorithmic}
\end{algorithm}

\begin{proof}
    We show correctness of \cref{alg:singleparam_small} and perform a complexity analysis.
    
   \textbf{Correctness.} We first claim that, with probability $\geq 1-\delta$, we have that 
   \begin{align}
       ||h_P|'-|h_P||&<\eps/4\label{eq:1learningsingleparameter},\\
       ||\widetilde h_P|'-|h_P+|h_P|'||&< \eps/4\label{eq:2learningsingleparameter}.
   \end{align} 
   
   Assume that claim for the moment. If $|\widetilde h_P|'\geq \eps/2,$ by \cref{eq:2learningsingleparameter}, we then have that $|h_P+|h_P|'|> \eps/4$. Now, by \cref{eq:1learningsingleparameter}, we have that have that $|h_P-|h_P|'|\leq \eps/4$, so $|h_P|'$ ($\eps/4$)-approximates $h_P$, as desired.  On the other hand, if $|\widetilde h_P|'< \eps/2,$ by \cref{eq:2learningsingleparameter} we have that $|h_P+|h_P|'|\leq 3\eps/4,$ so $-|h_P|'$ $(3\eps/4)$-approximates $h_P,$ as desired. 

   Now, we prove \cref{eq:1learningsingleparameter,eq:2learningsingleparameter}. By Taylor expansion, \cref{eq:Taylor}, we have that 
   \begin{equation*}
       |u_P-(-ith_P)|\leq Ct^2h_P^2,
   \end{equation*}
   where $C\geq 1$ is the constant in \cref{eq:Taylor}.
   Now, as $t=1/800\eps C$ and $|h_P|\leq 10\eps$, we have that 
   \begin{equation}\label{eq:3learningsingleparameter}
       |u_P-(-ith_P)|\leq \frac{t|h_P|}{80}.
   \end{equation}
   On the other hand, by the Hoeffding bound, \cref{lem:hoeffding}, we have that with probability $\geq 1-\delta/2$,
   \begin{equation}\label{eq:4learningsingleparameter}
       ||u_P|'-|u_P||< \frac{1}{6400C}.
   \end{equation}
   Putting \cref{eq:3learningsingleparameter,eq:4learningsingleparameter} together, and recalling that $|h_P|\leq 10\eps$ and $t=1/800C\eps$, we have that 
   \begin{equation*}
   \bigg||h_P|-\underbrace{\frac{|u_P|'}{t}}_{\coloneqq|h_P|'}\bigg|< \frac{|h_P|}{80}+ \frac{1}{6400Ct}\leq \frac{\eps}{8}+\frac{\eps}{8}=\frac{\eps}{4},   
   \end{equation*}
   which proves \cref{eq:1learningsingleparameter}. \cref{eq:2learningsingleparameter} is proved in the same way. 
   
   \textbf{Complexity analysis.} The number of experiments and the number of queries is $O(\log(1/\delta))$. Each query is done at time $t=\Theta(1/\eps),$ so the total evolution time is $O((1/\eps)\cdot\log(1/\delta))$ and the time resolution is $t=\Omega(1/\eps).$ The classical post processing time is $O(\log(1/\delta))$, as one just has to count the number of appearances of $P$ in $O(\log(1/\delta))$ repetitions of Pauli sampling. The number of extra qubits is $n$, which are those required to perform Pauli sampling. Note that, by \cref{eq:4learningsingleparameter}, we only need a $(1/6400C)$-error approximation of $|u_P|$, so we can allow for SPAM error of $(1/12800C)$ and perform Pauli sampling more times (but still constant) to obtain the approximation of \cref{eq:4learningsingleparameter}.
\end{proof}

Now, we drop the condition of $h_P$ being small enough via a binary search argument, which results in just a logarithmic overhead.

\begin{lemma}\label{lem:singleparameterlearningunconstrained}
    Let $P\in \{I,X,Y,Z\}^{\otimes n}$, let $\eps,\delta>0$, let $H=h_PP$ for some $|h_P|\leq 1$. Then, there is an algorithm that with $O((1/\eps)\cdot\log(\log(1/\eps)/\delta))$ total evolution time and $O(\log(\log(1/\eps)/\delta)\log(1/\eps))$ experiments outputs $h_P'$ such that $|h_P-h_P'|\leq \eps$ with probability $\geq 1-\delta.$

    The algorithm uses just $n$ extra qubits, it only makes $O((1/\eps)\cdot\log(\log(1/\eps)/\delta)\log(1/\eps))$ queries to the time evolution operator of $H,$ the time resolution is $\Omega(1)$, and it only uses $O(\log(\log(1/\eps)/\delta)\log(1/\eps))$ classical post-processing time, and is robust against a constant amount of SPAM error.
\end{lemma}

\begin{algorithm}[H]
\caption{Learning an isolated single Pauli coefficient $h_P$}
\label{alg:singleparamunconstrained}
\begin{algorithmic}[1]
\Require Pauli operator $P \in \{I, X, Y, Z\}^{\otimes n}$, error $\varepsilon > 0$, failure probability $\delta \in (0,1)$, Hamiltonian $H=h_PP$ where $|h_P| \leq 1$
\Ensure Estimate $h_P'$ such that $|h_P - h_P'| \leq \varepsilon$ with probability $\ge 1-\delta$
\State Set $L \gets \lceil\log_{10}(1/\varepsilon)\rceil$.
\State Initialize $H_1 \gets H$ and $h_P' \gets 0$.
\For{$l=1$ to $L$}
    \State Use algorithm \cref{alg:singleparam_small} with Hamiltonian $H_l$.
    \State Use parameters $\delta_l \gets \delta/L$ and $\varepsilon_l \gets 1/10^l$.
    \State Get estimate $h_l'$ of $h_l$ using the algorithm.
    \State Set $h_P' \gets h_P' + h_l'$.
    \State Set $H_{l+1} \gets H_l - h_l'P$.
\EndFor
\State \Return $h_P'$
\end{algorithmic}
\end{algorithm}

\begin{proof}
   We show correctness of \cref{alg:singleparamunconstrained} and perform a complexity analysis.
    
   \textbf{Correctness.} Note that the first iteration succeeds with probability $\geq 1-\delta/L$, because $|h_P|\leq 1=10\cdot (1/10^1).$ If the $l$-th iteration succeeds, then  $|h_l-h_l'|=|h_{l+1}|\leq 1/10^l=10\cdot(1/10^{l+1}),$ so one can perform the $(l+1)$-th iteration and and succeed with probability $\geq 1-\delta/L.$  Thus, with probability $\geq 1-\delta$ all iterations succeed and $|h_L-h_L'|\leq 1/10^L\leq\eps$. As $h_L=h_P-\sum_{i\in[L-1]}h_i'$, we have that $\sum_{i\in [L]}h_i'$ is an $\eps$-approximation of $h_P$ with probability $\geq 1-\delta.$ 

    \textbf{Complexity analysis.} The complexity analysis follows from applying the complexity analysis of Lemma \ref{lem:singleparameterlearningconstrained} to each of the $L$ iterations. The number of experiments is clearly $L\cdot O(\log (L/\delta))= O(\log(\log(1/\eps)/\delta)\cdot \log (1/\eps)) $. The total time evolution is $\sum _{l=1}^{L}O(1/\eps_l\cdot \log({L/\delta}))=O(\log (L/\delta)\cdot\sum_{l=1}^{L} 10^l)=O(\log (L/\delta)10^L)=O( \frac 1 \eps\cdot \log(\log(1/\eps)/\delta))$. The robustness against SPAM errors is inherited from Lemma~\ref{lem:singleparameterlearningconstrained}.
   \textbf{}
\end{proof}

Now, we are ready to prove the main result of this section.

\begin{theorem}\label{theo:singleparameterlearning}
        Let $H$ be an $n$-qubit $s$-sparse Hamiltonian. Let $P\in\{I,X,Y,Z\}^{\otimes n}$ with $|h_P|\leq 1$. Then, there is an algorithm that with $O( (1/\eps)\cdot\log(\log(1/\eps)/\delta))$ total evolution time and $O(\log(\log(1/\eps)/\delta)\log(1/\eps))$ experiments outputs $h_P'$ such that $|h_P-h_P'|\leq \eps$ with probability $\geq 1-\delta.$

        The algorithm uses just $n$ extra qubits, it only makes $\widetilde O(s/(\delta \eps^{5/2}\norm{H}_{\opnorm}^{3/2}))$ queries to the time evolution operator of $H,$ the time resolution is $\Omega(\delta\eps^{1/2}/(s\norm{H}_{\opnorm}^{3/2}))$, and it only uses $O(\log( \log(1/\eps)/\delta)\log(1/\eps))$ classical post-processing time.
\end{theorem}
\begin{proof}
       We first present the algorithm, then show correctness, and finally perform a complexity analysis.

       \textbf{The algorithm.} The algorithm consists of two steps:
        \begin{enumerate}
            \item Let $r=\lceil \log(2s/\delta)+2\rceil$. Sample $Q_1,\dots,Q_r$ iid uniform elements of $\{Q\in\{I,X,Y,Z\}^{\otimes n}:\ [Q,P]=0\}$.
            \item Obtain an estimate $h_P'$ of $h_P$ using the \cref{alg:singleparamunconstrained} with $H_{Q_1,\dots,Q_r}$ (treating $H_{Q_1,\dots,Q_r}$ as it was $h_PP$). 
        \end{enumerate}

        \textbf{Correctness.} We claim that the probability of $H_{Q_1,\dots,Q_r}=h_PP$ is greater or equal than $\geq 1-\delta/2$. In that case, by \cref{lem:singleparameterlearningunconstrained} we have that $|h_P'-h_P|\leq \eps$ with probability $\geq 1-\delta/2$, as desired. Thus, it only remains to prove the claim. We write 
        \begin{equation*}   
            \Pr[H_{Q_1,\dots,Q_r}=h_PP]=\Pr[P\text{ is the only element of }\mathcal H\text{ that commutes with all }Q_1,\dots,Q_r].        \end{equation*}
        By the same argument that we used in the proof of Lemma~\ref{lem:survival}, we may assume that $P=Z\otimes I^{\otimes (n-1)}$. Now, following the identification of the Pauli strings with $(2n)$-bit strings as done in the proof of Lemma \ref{lem:survival}, we can associate $(2n)$-bit strings $x^1,\dots,x^{s'}$ to the Paulis in  $\mathcal H-\{P\}$ ($s'=s$ if $P\notin \mathcal H$ and $s'=s-1$ otherwise), $x^0=e_{n+1}$ to $P=Z\otimes I^{\otimes (n-1)}$, and $(y^1,\dots,y^r)$ to $Q_1,\dots,Q_r$. Thus, we can write
        \begin{equation*}
            \Pr[H_{Q_1,\dots,Q_r}=h_PP]=\Pr[ \{\forall\  1\leq i\leq s',\ \exists\ j\ [x^i,y^j]=1\}|\{[x^0,y^j]=0\ \forall j\in [r]\}]. 
        \end{equation*}
        where $y^j$ are iid uniform elements of $\{0,1\}^{2n}$. Now, we use that $x^0=e_{n+1}$, so
        \begin{align*}
             \Pr[H_{Q_1,\dots,Q_r}=h_PP]&=\Pr[ \{\forall\  1\leq i\leq s',\ \exists\ j\ \sum_{l\neq  [n+1]}x^i_{l}y^j_l=1\}].
        \end{align*}
        Finally, note that $(x^i_1,\dots,x^i_{n},x^{i}_{n+2},\dots,x^i_{2n})\neq 0^{2n-1}$ for $i\in [s']$, because $x^i\neq x^0=e^{n+1}$ and $x^i\neq 0^{2n}.$ Thus, we an apply Corollary \ref{cor:ValiantVaziranidelta} to ensure that $\Pr[H_{Q_1,\dots,Q_r}=h_PP]\geq 1-\delta/2.$

        \textbf{Complexity analysis.} All the complexities, but the number of queries and the time resolution, are the same as the ones we would get in Lemma \ref{lem:singleparameterlearningunconstrained} in the case that we used Lemma \ref{lem:singleparameterlearningunconstrained} with $H_{Q_1,\dots,Q_r}=h_PP$. The number of queries and time resolution have an overhead due to the implementation of $e^{-itH_{Q_1,\dots,Q_r}}$ with queries to $e^{-itH}$. Tracing back to \cref{eq:4learningsingleparameter}, we notice that it suffices to approximate $e^{-itH_{Q_1,\dots,Q_r}}$ up to constant error and time $t= O(1/\eps)$. By Lemma \ref{lem:TrotterizationOfHQ}, this can be done with $O(2^r(\norm{H}_{\opnorm}/\eps)^{3/2})=O((s/\delta)\cdot(\norm{H}_{\opnorm}/\eps)^{3/2})$ queries to the time evolution operator of $H$ at time $\Omega(\delta\eps^{1/2}/(s\norm{H}_{\opnorm}^{3/2})).$ Thus the total number of queries is $O((s/\delta)\cdot(\norm{H}_{\opnorm}/\eps)^{3/2}(1/\eps)\cdot\log(\log(1/\eps)/\delta)\log(1/\eps))=\widetilde O(s/(\delta \eps^{5/2}\norm{H}_{\opnorm}^{3/2}))$ and the time resolution is $\Omega(\delta\eps^{1/2}/(s\norm{H}_{\opnorm}^{3/2})).$ The robustness against SPAM errors is inherited from Lemma~\ref{lem:singleparameterlearningconstrained}.
\end{proof}
\subsection{Putting everything together}
We are now ready to prove the upper bound for learning sparse Hamiltonians. 
\begin{theorem}\label{theo:learninginellinftynorm}
        Let $H$ be an $n$-qubit $s$-sparse Hamiltonian with $|h_P|\leq 1$ for every $P\in\{I,X,Y,Z\}^{\otimes n}$. Then, there is an algorithm that with $\widetilde O( (s/\eps)\cdot\log(1/\delta))$ total evolution time and $\widetilde O(s\log(1/\delta)\log(1/\eps))$ experiments outputs $H'$ such that $\mathcal H'\subseteq \mathcal H$ and $|h_P-h_P'|\leq \eps$ for every $P\in\{I,X,Y,Z\}^{\otimes n}$ with probability $\geq 1-\delta.$

        The algorithm uses just $n$ extra qubits, it only makes $\widetilde O(s^2/(\delta \eps^{5/2}\norm{H}_{\opnorm}^{3/2}))$ queries to the time evolution operator of $H,$ the time resolution is $\Omega(\delta\eps^{1/2}/(s\norm{H}_{\opnorm}^{3/2}))$, and it only uses $O(ns\log(s/\delta) \cdot \log(1/\eps))$ classical post-processing time, and it is robust against a constant amount of SPAM errors.
\end{theorem}
\begin{proof}
   We first present the algorithm, then show correctness, and finally perform a complexity analysis.

   \textbf{The algorithm.} The algorithm has three stages. 
   \begin{enumerate}
       \item Use \cref{alg:theorem34} to find $\mathcal S$ such that $\mathcal H_\eps\subseteq \mathcal S$ and $|\mathcal S|=\widetilde O(s)$ with probability at least $1-\delta/2.$
       \item For every $P\in\mathcal S$, run the algorithm of \cref{theo:singleparameterlearning} with $\delta'=\delta/2|\mathcal S |$ and $\eps'=\eps/2$, to obtain an estimate $h'_P$ of $h_P$.
       \item For every $P$, define $h_P''=0$ if $|h_P'|\leq \eps/2$ and $h''_P=h_P'$ otherwise. 
   \end{enumerate}
   Output $H''=\sum_{P\in\mathcal P}h_P''P.$

   \textbf{Correctness.} From \cref{theo:singleparameterlearning,theo:structurelearning} and a union bound follows that, with probability $\geq 1-\delta$, $\mathcal P$ contains $\mathcal H_\eps$ and $|h_P-h_P'|\leq \eps/2$ for every $P\in\mathcal P$. 
   
   Now, we show that $|h_P-h_P''|\leq \eps$ for every $P\in\{I,X,Y,Z\}^{\otimes n}$. If $P\notin\mathcal P,$ then $|h_P|\leq \eps$ and $h_P''=h_P'=0,$ so $|h_P''-h_P|\leq \eps.$ If $P\in\mathcal P$, then we have both $|h_P-h_P'|\leq \eps/2$ and $|h'_P-h_P''|\leq \eps/2$, so by triangle inequality $|h_P''-h_P|\leq \eps$.
   
   Finally, note that $\mathcal H''\subseteq \mathcal H$, because if $h_P=0$ and $P\in\mathcal P$, then $|h_P'|\leq \eps/2$ by \cref{theo:singleparameterlearning}, so by definition $h''_P=0$; and if $h_P=0$ and $P\notin\mathcal P$, then $h_P'=0,$ so also $h_P''=0.$ 

   \textbf{Complexity analysis.} The complexities follow from the complexities of invoking once the algorithm of \cref{theo:structurelearning} and $\widetilde O(s)$ times the algorithm of \cref{theo:singleparameterlearning}.
\end{proof}

\begin{corollary}\label{cor:learninginellinoperatornom}
     Let $H$ be an $n$-qubit $s$-sparse Hamiltonian with $|h_P|\leq 1$ for every $P\in\{I,X,Y,Z\}^{\otimes n}$. Then, there is an algorithm that with $\widetilde O( (s^2/\eps)\cdot\log(1/\delta))$ total evolution time and $\widetilde O(s\log(1/\delta)\log(1/\eps))$ experiments outputs an $s$-sparse Hamiltonian $H'$ such that $\sum_{P}|h_P-h_P'|\leq \eps$ (which implies $\norm{H-H'}_{\opnorm}\leq\eps$) with probability $\geq 1-\delta.$
\end{corollary}
\begin{proof}
   We first present the algorithm, then show correctness, and finally perform a complexity analysis.

   \textbf{The algorithm.} The algorithm is the algorithm of \cref{theo:learninginellinftynorm} with $\eps'=\eps/s.$ Let $H'$ be the output of this algorithm.

   \textbf{Correctness.} By \cref{theo:learninginellinftynorm}, with probability $\geq 1-\delta$ we have that $\mathcal H'\subseteq \mathcal H$ and that
   \begin{equation*}
       |h_P-h_P'|\leq \eps/s
   \end{equation*}
   for every $P\in\mathcal H.$ In that case, 
   \begin{equation*}
      \sum_{P\in\{I,X,Y,Z\}^{\otimes n}}|h_P-h_P'|=\sum_{P\in\mathcal H}|h_P-h_P'|\leq\sum_{P\in\mathcal H}\frac{\eps}{s}\leq \eps.
   \end{equation*}
   As $\norm{H-H'}_{\opnorm}\leq \sum|h_P-h_P'|$ by triangle inequality, we are done.

   \textbf{Complexity analysis.} The complexity follows immediately from \cref{theo:learninginellinftynorm} with $\eps'=\eps/s.$
\end{proof}

Combining Corollary~\ref{cor:learninginellinoperatornom} with Lemmas~\ref{lem:upperboundtodT} and \ref{lemma_dB_upperbound}, we have the following result for learning Hamiltonians in physically motivated distances.

\begin{corollary}[Learning in physically  motivated distances]\label{cor:learninphysicaldistances}
    Let $H$ be an $n$-qubit $s$-sparse Hamiltonian with $|h_P|\leq 1$ for every $P\in\{I,X,Y,Z\}^{\otimes n}$. Let $T,B> 0.$ Then, there is an algorithm that with $\widetilde O( (s^2/\eps)\cdot\log(1/\delta))$ total evolution time and $\widetilde O(s\log(1/\delta)\log(1/\eps))$ experiments outputs an $s$-sparse Hamiltonian $H'$ such that $d_T(H,H')\leq T\eps$ and $d_B(H,H')\leq B \eps$, with probability $\geq 1-\delta.$
\end{corollary}

\subsection{Modifications to remove the ancilla qubits}\label{sec:memoryless}
In the algorithm of \cref{theo:learninginellinftynorm} the only reason why we use $n$ ancilla qubits is to perform Pauli sampling, whose naive implementation requires $n$ ancilla qubits. However, in our algorithm we morally perform Pauli sampling only of unitaries of the kind $U=u_I I+u_P P$, where $P$ is an unknown Pauli coefficient. In this section, we show that one can simulate Pauli sampling for these unitaries without the need of ancilla qubits, and thus we obtain a memory-less version of \cref{theo:learninginellinftynorm}.

\begin{lemma}\label{fact:XZindepsampling}
    Let $P=\alpha \bigotimes_{i=1}^{n}X^{{\mathbf a_i}}Z^{{\mathbf b_i}}\neq I$ be an $n$ qubit Pauli operator. Let $U=u_I I+u_P P$ be a unitary. 
    
    Then, if $\mathbf a\neq 0,$ the outcome of measuring $U\ket{0}^{\otimes n}$ in the $Z$-basis is $\mathbf a$ with probability $|u_P|^2$, and $0^n$ with probability $|u_I|^2$. If $\mathbf a=0^n,$ the outcome of that measurement is always $0^n.$
    
    Similarly, if $\mathbf b\neq 0^n$, the outcome of measuring $U\ket{+}^{\otimes n}$ in the $X$-basis is $\mathbf b$ with probability $|u_P|^2$, and the $0^n$ with probability $|u_I|^2$. If $\mathbf b=0^n,$ the outcome of that measurement is always $0^n.$
\end{lemma}

\begin{proof}
    We only prove the first part, as the second follows via symmetry between $X$ and $Z$. If $\mathbf a=0^n$, then $U\ket{0}^{\otimes n}\propto \ket 0^{\otimes n}$, so we sample the bit vector $0^n=\mathbf a$ with probability $1$, as claimed. If $\mathbf a\neq 0^n$, we note that $P\ket{0}^{\otimes n}\propto \ket {\mathbf a}$, as the $Z$ factors only contribute with a global phase. Hence, $U\ket{0}^{\otimes n}=u_I \ket 0^{\otimes n}+\mathrm{e}^{i\phi}u_P\ket{\mathbf a}$, for some $\phi\in\R$. Since $\ket{\mathbf a}$ and $\ket{0}^{\otimes n}$ are orthogonal, measuring $U\ket{0}^{\otimes n}$ in the computational basis leads to the distribution as per the theorem statement.
\end{proof}

\begin{remark}
    We claim that \cref{fact:XZindepsampling} suffices to reproduce \cref{theo:learninginellinftynorm} without ancilla qubits. Indeed, in \cref{theo:learninginellinftynorm} we only use Pauli sampling for unitaries of the kind $U=u_II+u_PP$.\footnote{This is not fully true, as when the isolation technique in \cref{alg:theorem34} is not successful, we do not have $H_{Q_1,\dots,Q_r}=h_PP$, and thus the correspondent unitary time evolution may not look like $U=u_II+u_PP$. However, in that case the performance of Pauli sampling is not relevant for our purposes, as all the elements of the effective support will be gathered in the iterations where the isolation technique succeed, and we do not suffer from gathering extra elements.} In particular, we only perform it with two concrete purposes: $i)$ to identify $P$ with constant probability in the case that $|u_P|=\Omega(1)$ (this is the use of Pauli sampling in \cref{theo:structurelearning}), $ii)$ to estimate $|u_P|$, when $P$ is known, up to constant precision (this is the use in \cref{sec:singleparameterlearning}). For both purposes, one can use \cref{fact:XZindepsampling}, that allows us to sample from $(\mathbf a,\mathbf b)$ for the case $\mathbf a,\mathbf b\neq 0^n$ with probability $|u_P|^4$, by measuring $U\ket{0}^{\otimes n}$ in the $Z$-basis and $U\ket{+}^{\otimes n}$ in the $X$-basis. For the case where $\mathbf a=0^n$ (or $\mathbf b=0^n$), we can sample $(\mathbf a,\mathbf b)$ with probability $|u_P|^2$, and we can identify this case by notiicing that the first measurement always outputs $0^n$. 
\end{remark}

\section{Lower bounds for learning}
\subsection{Number of experiments lower bound}

In this section, we prove lower bounds on the number of experiments  for learning an $s$ sparse Hamiltonian in the operator norm. We provide two lower bounds. The first one is  $\Omega((s/n)\cdot \log (1/\eps))$, which holds even when the $s$ Pauli operators are already known. The second one is $\Omega(s)$ for the special case when $\eps\leq  1/s$.

\begin{theorem}\label{thm:lb_experiments_1}
     Let $\mathcal H$ be a subset of $\{I,X,Y,Z\}^{\otimes n}-\{I^{\otimes n}\}$ with size $s$. Assume that there is an algorithm with $O(n)$ ancilla qubits that learns all Hamiltonians with operator norm at most 1 whose Pauli support lies in $\mathcal H,$ up to error $\eps$ in operator norm with success probability at least 0.9. Then, the algorithm makes at least $\Omega((s/n)\cdot \log(1/\eps))$ experiments.
\end{theorem}

\begin{proof}
    We identify the space of Hamiltonians whose Pauli support lies in $\mathcal H$ with $\mathbb R^{s}$. For any $H\in \R^s$, define $B(H,\eps)= \{H'\in \R^s: \norm{H-H'}_{\opnorm}\leq \eps\}$. Let $\mathscr H$ be a subset of $\{H\in \R^s:\, \norm{H}_{\opnorm}\leq 1\}$ with maximum size such that every pair of Hamiltonians in $\mathscr H$ are $2\eps$-far in operator norm. We will show a lower bound on the size of $\mathscr H$ by a volume argument.

    Let $\mu$ be the Lebesgue measure in $\R^s$. Then,
    \begin{equation}\label{eqn:lebesgue_ball}
        \mu(B{(H,2\eps)})=\mu(B(0,2\eps))=(2\eps)^s \mu(B(0,1))
    \end{equation}

    \noindent Since $\mathscr H$ is of maximum size, we know that there is no $H'\in \R^s$ with $\norm{H'}_{\opnorm}\leq 1$ such that $H'$ is $(2\eps)$-far from all $H\in \mathscr H$. This means that $$ B(0,1)\subseteq \cup_{H\in\mathscr H}B(H,2\eps).$$ As $\mu$ is a measure, we have that
    \begin{equation}\label{eq:lowerboundonmathscrH}
        \sum_{H\in \mathscr H}\mu( B(H,2\eps))\geq \mu(B(0,1))\quad 
        \underset{\eqref{eqn:lebesgue_ball}}{\implies}\quad |\mathscr H|\geq \left(\frac 1 {2\eps}\right)^s.
    \end{equation}

    Now, we claim that any algorithm to learn Hamiltonians from $\mathscr H$ with $N$ experiments and with $O(n)$ ancilla qubits with success probability 0.9, determines an algorithm to
    win the input guessing game of $K=\lfloor (1/2\eps)^s\rfloor$ with success probability 0.9. Indeed, before playing the game, by \cref{eq:lowerboundonmathscrH}, Alice and Bob can agree on an encoding of $[K]$ into the elements of $\mathscr H$. Then, when Alice receives the input $x\in [K]$, she can run the $N$ experiments algorithm for learning the Hamiltonian $H_x$ corresponding to $x$. This produces $O(Nn)$ bits, which she can send to Bob, who can post-process it and determine a Hamiltonian $H'$, which will be $\eps$-close to $H_x$ with probability at least 0.9. As all the elements of $\mathscr H$ are $(2\eps)$-far to each other, we have that $x$ is the only element of $[K]$ such that its corresponding Hamiltonian is $\eps$-close. Thus, Bob can determine $x$ with probability at least 0.9. Finally, by \cref{lem:Nayaksbound} we have that 
    \begin{equation*}
        O(Nn)=\Omega(s\log(1/\eps)), 
    \end{equation*}
    so $N=\Omega((s/n)\cdot \log(1/\eps)),$ as claimed. 
\end{proof}

Now, we prove our second lower bound.
\begin{theorem}\label{thm:lb_experiments_2}
    Let $\eps\leq 1/(2s)$. Consider an algorithm with $O(n)$ ancilla qubits that $\eps$-learns in operator norm all $n$-qubit $s$-sparse Hamiltonians with operator norm at most 1, with success probability $\geq 0.9$. Then, the algorithm must make $\Omega(s)$ experiments.
\end{theorem}
\begin{proof}
Let $\mathscr H=\{\sum_{P\in \mathcal H} 2\eps P:\, \mathcal H\in\{I,X,Y,Z\}^{\otimes n}-\{I^{\otimes n}\},\,|\mathcal{H}|=s\}$. Notice that as $s=O(2^n)$, $$|\mathscr H|=\binom{4^n}{s}=\Omega(4^{ns}).$$ Note that, as $\eps\leq 1/(2s)$, we have that every Hamiltonian of $\mathscr H$ has operator norm at most 1.  Also, by Lemma \ref{lem:comp_linf_opnorm}, any two such Hamiltonians differ in operator norm by at least $2\eps$. 

Now, by the same argument via the the input guessing game used in the proof of \cref{thm:lb_experiments_1}, we have that the number of experiments $N$ made by the algorithm must satisfy 
\begin{align*}
    O(Nn)=\Omega\left(\log(4^{ns})\right)=\Omega(ns).
\end{align*}
Hence, $N=\Omega(s),$ as claimed. 
\end{proof}
Combining \cref{thm:lb_experiments_1,thm:lb_experiments_2,lemma_2}, we get the following lower bound for learning in the time-constrained distance.
\begin{corollary}\label{cor:lb_experiments_timeconstrained}
    Let $T>0$ be a constant. Assume that there is an algorithm with $O(n)$ ancilla qubits that learns all $n$-qubit $s$-sparse Hamiltonians $H$ with $d_T(0,H)\leq 1$, up to error $\eps$ in time-constrained distance, with success probability at least 0.9. Then, the algorithm makes at least $\Omega((s/n)\cdot \log(1/\eps))$ experiments. Furthermore, if $\eps\leq 1/(2s)$, the number of experiments must be at least $\Omega(s).$
\end{corollary}
\begin{remark}\label{rem:lb_experiment_c1}
\cref{thm:lb_experiments_1}  shows a lower bound for learning Hamiltonians with $\norm{H}_{\opnorm}\leq 1$ within operator norm-distance. However, given any other norm $\|\cdot\|$, the same argument shows the same lower bound for learning Hamiltonians with $\norm{H}\leq 1$ within the distance determined by $\norm{\cdot}.$

For the lower bound in \cref{thm:lb_experiments_2}, we additionally need $\norm{\cdot}$ to satisfy $\|P\|=1$ for every Pauli string $P$ and  $\norm{H}\geq \norm{H }_{\ell_\infty}$. In particular, the lower bound applies to the $\ell_p$ norms of the Pauli coefficients.
\end{remark}

\subsection{Total time evolution lower bound}

\begin{theorem}\label{thm:lb_timeevol_c1}
    Consider an algorithm that $\eps$-learns in $\ell_1$ norm all $n$-qubit $s$-sparse Hamiltonians $H$ with $\norm{H}_{\ell_1}\leq 1$, with success probability $\geq 0.9$. Then, the algorithm must use $\Omega(s/\eps)$ evolution time.
\end{theorem}
\begin{proof}
    Fix any set $\mathcal H\subseteq \{I,X,Y,Z\}^{\otimes n}$ with size $s$. Let $\mathscr H=\{\sum_{P\in\mathscr H} h_PP: h_P=\pm \eps/s\}$. Let $\mathcal A$ be an algorithm that $(\eps/10)$-learns $s$-sparse Hamiltonians in $\ell_1$ norm with success probability $\geq 0.9.$ We define another algorithm $\mathcal A'$ that  $(\eps/5)$-learns every element of $\mathscr H$ in $\ell_1$ norm and, additionally, outputs an element of $\mathscr H$: run $\mathcal A$ and round the output to the the closest element $H'$ of $\mathscr H$. Note that total evolution time and success probability remain the same.

    Now, we make the following claim. 
    \begin{claim}\label{claim:lb}
        There is $P_0\in\mathcal H$ such that $\mathcal A'$, when run on a uniformly random Hamiltonian from $\mathscr H$, learns the corresponding Pauli coefficient exactly with probability at least $0.81.$
    \end{claim}
\noindent  \emph{Proof of the \cref{claim:lb}.} 
        Let $H'$ be the output of $\mathcal A'$. As, with probability at least 0.9, $\norm{H-H'}_{\ell_1}\leq \eps/5$,  and $H,H'\in \mathscr H$; we have that a 0.9$s$ of the Pauli coefficients of $H$ and $H'$ are equal, with probability at least $0.9$. Now, if we define $X_P$ as the random variable that takes value 1 if the Pauli coefficient corresponding to $P$ for $H$ equals the one $H'$, we have that 
        \begin{align*}
            \mathbb E_{H\in\mathscr H,\ \text{randomness of }\mathcal A'}\left[\sum_{P\in\mathscr H}X_P\right]\geq 0.9\cdot 0.9s=0.81s.
        \end{align*}
        By linearity of expectation, we have that there is $P\in \mathcal H$ such that $\mathbb E[X_P]\geq 0.81.$
    \qed

    Finally, we claim that $\mathcal A'$ can be used to distinguish the Hamiltonians to solve the task of distinguish $(\eps/s) P_0$ and $-(\eps/s)P_0$, when having access to their time evolution operators and when both are drawn with probability $1/2$. Proving that claim is enough to conclude the proof, because, learning a single Pauli coefficient of a Hamiltonian up to error $\eps'$ requires $\Omega(1/\eps')$ total evolution time. This lower bound on single parameter learning is well-known (see for instance \cite{Dutkiewicz.2023,kallaugher2025hamiltonian}), but we show it for completeness. By subadditivity of diamond norm, the fact that $$\norm{e^{-it(\eps/s)P}(\cdot) e^{it(\eps/s)P}-e^{it(\eps/s)P}(\cdot) e^{-it(\eps/s)P}}_{\diamond }=O\left( \frac{t\eps}{s}\right),$$ and the operational interpretation of the diamond norm (see \cref{eq:physicalinterpretationofdiamond}), we have that the algorithm must use $\Omega(s/\varepsilon)$ total evolution time.  

    Thus, it remains that $\mathcal A'$ can be turned into a distinguisher for $(\eps/s) P_0$ and $-(\eps/s)P_0$. We are given time-evolution access to $H_0$, that equals $\pm  \eps/sP_0$ with probability $1/2$. If we pick  $H_1\in \mathscr H_1=\{\sum_{P\in\mathscr H-\{P_0\}} h_PP: h_P=\pm \eps/s\}$, we have that $H=H_0+H_1$ is a uniformly random element of $\mathscr H.$ Thanks to Trotterization (see \cref{theo:trotterization}), we can have access to $e^{-itH}$ for a uniformly random $H\in\mathscr H$, without any time-evolution overhead. Hence, it follows from \cref{claim:lb} that $\mathcal A'$ allows to distinguish between $(\eps/s)P_0$ and $-(\eps/s)P_0$ with probability $\geq 0.81.$ 
    \end{proof}

The following lower bound follows from using the same proof strategy as in \cref{thm:lb_timeevol_c1}. 
\begin{theorem}\label{theo:lbtimeinopnorm}
    Consider an algorithm that $\eps$-learns in operator norm all $n$-qubit $s$-sparse Hamiltonians $H$ with $\norm{H}_{\opnorm}\leq 1$, with success probability $\geq 0.9$. Then, the algorithm must use $\Omega(\sqrt{s}/\eps)$ evolution time. Furthermore, if $s\leq n$, then the algorithm must use $\Omega(s/\eps)$ evolution time. 
\end{theorem}
\begin{proof}
    For the lower bound of $\Omega(\sqrt{s}/\eps)$, we can mimic the proof of Theorem \ref{thm:lb_timeevol_c1}, choosing the $\ell_1$ error to be $ \eps\sqrt s$, because for two $s$-sparse Hamiltonians $H$ and $H'$, we have that $\norm{H-H'}_{\ell_1}\leq \sqrt{2s}\norm{H-H}_{\opnorm}$ by \cref{lem:comp_l1_opnorm}.

    For the case when $s\leq n$, we define $\mathcal H=\{I\otimes \dots\otimes \underbrace{Z}_{i\text{-th qubit}}\otimes \dots \otimes I:\, i\in [s]\}.$
    Now, consider $a:\mathcal H\to \{0,1\}.$ We define $\ket{\psi_a}:=\left(\bigotimes_{i=1}^s\ket{a(Q_i)}\right)\otimes \ket{0}^{\otimes {n-s}}.$ We have that, for every $P\in\mathcal H$, 
    \begin{equation}\label{eq:lbauxiliary}
        \bra{\psi_a}P\ket{\psi_a}=(-1)^{a(P)}.
    \end{equation}
    Next, we claim that for every $H=\sum_{P\in\mathcal H}h_PP$, we have that $\norm{H}_{\opnorm}\geq \norm{H}_{\ell_1}$. Indeed, define 
    \[a(P)=\begin{cases}
       1 & \text{if $h_P<0$},\\
       0 & \text{otherwise}.
    \end{cases}\]
    \noindent Then, by \cref{eq:lbauxiliary},
    \[\norm{H}_{\opnorm}\geq \bra{\psi_a} H\ket{\psi_a}=\sum_{P\in \mathcal H} h_p \bra{\psi_a}P\ket{\psi_a}=\sum_{P\in \mathcal H}h_p(-1)^{a(P)}=\sum_{P\in \mathcal H}|h_P|=\norm{H}_{\ell_1},\]
    as claimed. Thus, for Hamiltonians supported on $\mathcal H$ we have that $\norm{H}_{\opnorm}=\norm{H}_{\ell_1},$ so mimicking the proof of \cref{lem:comp_l1_opnorm} we arrive at $\Omega(s/\eps)$ time evolution lower bound for learning in the operator norm. 

\end{proof}

Combining \cref{theo:lbtimeinopnorm,lemma_2}, we obtain the following. 
\begin{corollary}\label{cor:lb_timeevol_timeconstrained}
    Let $T>0$ be a constant. Consider an algorithm that $\eps$-learns in operator norm all $n$-qubit $s$-sparse Hamiltonians $H$ with $d_T(0,H)\leq 1$, with success probability $\geq 0.9$. Then, the algorithm must use $\Omega(\sqrt{s}/\eps)$ evolution time. Furthermore, if $s\leq n$, then the algorithm must use $\Omega(s/\eps)$ evolution time.
\end{corollary}

\section{Physically  motivated distances between Hamiltonians}\label{sec:physcallymotivateddistances}
In this section, we first introduce notions of distance between Hamiltonians that carry a strong physical meaning. Subsequently, we show that these notions are closely related to the operator norm between Hamiltonians, thereby endowing the latter with a clear operational interpretation.  Before proceeding, it is useful to recall the most relevant notions of distance for quantum states and quantum channels.

For quantum states, the most natural distance between two density operators $\rho_{1}$ and $\rho_{2}$ is the \emph{trace distance}, defined as
\begin{equation*}
    \frac{1}{2}\|\rho_{1} - \rho_{2}\|_{\trnorm}\,.
\end{equation*}
The operational meaning of the trace distance is captured by the well-known \emph{Holevo-Helstrom theorem} \cite{HOLEVO1973337,Helstrom}. This result establishes that if one is given one copy of an unknown quantum state, promised to be either $\rho_{1}$ or $\rho_{2}$ with equal prior probability $1/2$, then the maximum probability of correctly identifying the state is
\begin{equation*}
    P_{\text{succ}} = \tfrac{1}{2}\Bigl(1 + \tfrac{1}{2}\|\rho_{1} - \rho_{2}\|_{\trnorm}\Bigr)\,.  
\end{equation*}
Moreover, the trace distance plays a crucial role because it directly bounds the distinguishability of expectation values of bounded observables. Specifically, for any observable $O$, Holder's inequality implies that~\cite{Wilde2011FromCT}
\begin{equation*}
    \big|\Tr[\rho_{1} O]-\Tr[\rho_{2} O]\big| \leq \|O\|_{\opnorm}\,\|\rho_{1}-\rho_{2}\|_{\trnorm}\,,
\end{equation*}
which shows that if two states are close in trace distance, then the expectation values of any bounded observable with respect to those states must also be close.

For quantum channels, the most meaningful notion of distance between two quantum channels $\Phi_1$ and $\Phi_2$ is the \emph{diamond distance}~\cite{aharonov1998,Sacchi2005,Wilde2011FromCT}, defined as
\begin{equation*}
    \frac{1}{2}\|\Phi_1 - \Phi_2\|_{\diamond}\coloneqq \sup_{\rho_{AA'}}\frac12\|\Id_{A}\otimes\Phi_1(\rho_{AA'})-\Id_{A}\otimes\Phi_2(\rho_{AA'})\|_{\trnorm}\,,
\end{equation*}
where the supremum is taken over all the input states $\rho_{AA'}$ on the input system $A'$ and an arbitrary ancilla $A$. The operational meaning of the diamond distance again follows from the Holevo-Helstrom theorem: if one is given a single use of an unknown channel, promised to be either $\Phi_1$ or $\Phi_2$ with equal prior probability $1/2$, then the maximum probability of correctly identifying the channel is
\begin{equation}\label{eq:physicalinterpretationofdiamond}
    P_{\text{succ}} = \tfrac{1}{2}\Bigl(1 + \tfrac{1}{2}\|\Phi_1 - \Phi_2\|_{\diamond}\Bigr)\,. 
\end{equation}
Moreover, if two channels are close in diamond distance, the expectation values of any observable on the output states $\Phi_1(\rho)$ and $\Phi_2(\rho)$ are also close, for all input states $\rho$. 
In the special case of unitary channels, Aharonov et al.~\cite{aharonov1998} proved that the diamond distance can be evaluated without the need of an ancillary system, and the optimization can be always restricted to pure input states. More precisely, given two unitaries $V$ and $W$, it holds that
\begin{equation}\label{eq:aharonov}
    \frac12\|V(\cdot) V^\dagger-W(\cdot) W^\dagger\|_{\diamond}=\sup_{\psi} \frac12\bigl\| V\ket{\psi}\!\bra{\psi}V^\dagger-W\ket{\psi}\!\bra{\psi}W^\dagger \bigr\|_{\trnorm}\,,
\end{equation}
and, by using the well-known formula for the trace distance between pure states, we can write
\begin{equation}\label{eq_diamond_pure1}
    \frac12\|V(\cdot) V^\dagger-W(\cdot) W^\dagger\|_{\diamond}=   \sqrt{1-\inf_{\psi}|\bra{\psi}V^\dagger W\ket{\psi}|^2}\,.
\end{equation}
Motivated by these meaningful distances for states and channels, we now turn to Hamiltonians. One may then ask: in which physical scenarios do Hamiltonians arise? In practice, there are two natural settings:
\begin{itemize}
    \item \emph{Time-evolution setting}, where the Hamiltonian governs the unitary dynamics of the system;
    \item \emph{Thermodynamics setting}, where the Hamiltonian determines the Gibbs state at the thermal equilibrium.
\end{itemize}
These two perspectives motivate the introduction of two notions of distance that capture Hamiltonian distinguishability in an operationally meaningful way: the \emph{time-constrained diamond distance} and the \emph{temperature-constrained trace distance}, which will be presented in the following subsections.  

\subsection{Time-constrained diamond distance}

We define the \emph{time-constrained diamond distance} between Hamiltonians at time $T$ as
\begin{equation}\label{eq:time_constr_diamond}
    d_T(H_{1},H_{2}) \coloneqq \tfrac{1}{2}\sup_{t\in[0,T]} \|\mathcal{U}_{H_1,t} - \mathcal{U}_{H_2,t}\|_{\diamond},
\end{equation}
where $\mathcal{U}_{H_j,t}$, for $j=1,2$, denotes the channel associated with the time evolution operator $e^{-iH_j t}$, namely $\mathcal{U}_{H_j,t}(\,\cdot\,) \coloneqq e^{-iH_j t}(\,\cdot\,)e^{iH_j t}.$ By the exposition above the time-constrained distance quantifies the best probability of successfully discriminating the Hamiltonians by optimising over all protocols of the form ``prepare, evolve for a time $t\leq T$, and finally measure". 

Here, the parameter $T$ can be understood as an experimental time budget. Moreover, if this time budget was unconstrained, then any pair of Hamiltonians would be maximally distinguishable.

\begin{remark}[Time-unconstrained diamond distance and motivation for $d_T(H_1,H_2)$]

Taking the limit $T \to \infty$ in Eq.~\eqref{eq:time_constr_diamond} yields a \emph{time-unconstrained} diamond distance which is overly sensitive to small perturbations of the Hamiltonians. 
Indeed, there exist pairs of Hamiltonians that are arbitrarily close in operator norm yet maximally distant with respect to this measure.  

To illustrate this, consider commuting Hamiltonians, i.e.~$[H_1,H_2]=0$. By using  Eq.~\eqref{eq_diamond_pure1} we obtain
\begin{equation*}
    d_\infty(H_1,H_2) = \sup_{t \ge 0} \sqrt{1 - \inf_{\ket{\psi}} \bigl| \bra{\psi} e^{-i(H_1-H_2)t} \ket{\psi} \bigr|^2 }.
\end{equation*}
If $H_1 - H_2$ is proportional to the identity, the two evolutions differ only by a global phase, and therefore $d_\infty(H_1,H_2) = 0$.  

If instead $H_1 - H_2$ is not proportional to the identity, then $d_\infty(H_1,H_2) = 1$. 
Indeed, in this case $H_1 - H_2$ has at least two distinct eigenvalues, say $\varepsilon_0$ and $\varepsilon_1$, with eigenvectors $\ket{0}$ and $\ket{1}$. 
Restricting the optimization to the subspace $\mathrm{Span}\{\ket{0},\ket{1}\}$ gives
\begin{equation}\label{eq:overlap_Hcomm}
\begin{split}
    \inf_{\ket{\psi}} \bigl| \bra{\psi} e^{-i(H_1-H_2)t}\ket{\psi} \bigr|^{2} 
    & \le \inf_{\ket{\psi}\in \mathrm{Span}\{\ket{0},\ket{1}\}} 
    \bigl| \bra{\psi} e^{-i(H_1-H_2)t}\ket{\psi} \bigr|^{2} \\
    & = \inf_{\ket{\psi}\in \mathrm{Span}\{\ket{0},\ket{1}\}} 
    \bigl| \bra{\psi} e^{-i(\varepsilon_{0} - \varepsilon_{1})\tfrac{\sigma_{z}}{2}t}\ket{\psi} \bigr|^{2}.
\end{split}
\end{equation}
Choosing, for example, $\ket{\psi}= (\ket{0}+\ket{1})/\sqrt{2}$, the overlap in Eq.~\eqref{eq:overlap_Hcomm} reduces to 
\[
    \bigl|\cos\bigl((\varepsilon_{0}-\varepsilon_{1})t/2\bigr)\bigr|,
\]
which vanishes whenever
\begin{equation*}
    \frac{(\varepsilon_0 - \varepsilon_1) t}{2} = \frac{\pi}{2}, \frac{3\pi}{2}, \dots,
\end{equation*}
corresponding to evolution into an orthogonal state within this two-dimensional subspace. 
Hence the minimum overlap is zero, and we conclude that $d_\infty(H_1,H_2)=1$.  

In summary, for commuting Hamiltonians $H_1,H_2$ one has
\begin{equation*}
    d_{\infty}(H_1,H_2) = 
    \begin{cases}
      0 & \text{if } H_1-H_2 \propto \Id, \\
      1 & \text{otherwise}.
    \end{cases}
\end{equation*}
Thus, in the infinite-time limit any nontrivial difference between Hamiltonians implies maximal distinguishability under the time-unconstrained diamond distance. 
This makes the latter both mathematically unsatisfactory and physically unrealistic, since in practice experiments can only probe finite evolution times. 
These observations motivate the introduction of the \emph{time-constrained diamond distance} in Eq.~\eqref{eq:time_constr_diamond}, which is operationally meaningful and, as shown in Subsections~\ref{sec_sub1} and \ref{sec_sub2}, closely related to the operator norm between Hamiltonians---in contrast to its unconstrained counterpart.
\end{remark}
We will now show that the time-constrained diamond distance $d_T(H_1,H_2)$ is equivalent to the operator norm $\|H_1 - H_2\|_{\opnorm}$ up to multiplicative factors depending on $T$. To begin, we establish an upper bound on $d_T(H_{1},H_{2})$.

\subsubsection{Upper bounding the time-constrained distance}\label{sec_sub1}
This subsection is devoted to establish an upper bound on the time-constrained diamond distance in terms of the operator norm of the difference between the Hamiltonians. Specifically, the following bound holds.
\begin{lemma}\label{lem:upperboundtodT}
    Let $H_1,H_2$ be Hamiltonians. Then
    \begin{equation}\label{eq_upp_not_known}
        d_T(H_{1},H_{2}) \le \sin \Big( \min \Big( \frac{\pi}{2}, T\Vert H_{1} - H_{2}\Vert_{\opnorm} \Big)\Big),
    \end{equation}
    and, in particular, the following weaker upper bound holds:
    \begin{equation}\label{eq_upp_known}
        d_T(H_{1},H_{2}) \le T \Vert H_{1} - H_{2}\Vert_{\opnorm}\,.
    \end{equation}
\end{lemma}
To begin with, we show how the bound in Eq.~\eqref{eq_upp_known} can be obtained from standard tools. Consider the definition of $d_T(H_1,H_2)$ in terms of the diamond norm: 
\begin{equation*}
    d_T(H_1,H_2)\coloneqq \sup_{t\in[0,T]}\frac{1}{2}\|e^{-iH_1t}(\,\cdot\,)e^{iH_1t}-e^{-iH_2t}(\,\cdot\,)e^{iH_2t}\|_\diamond\,.
\end{equation*}
We can bound the above expression using standard tools as follows:
\begin{equation*}
    \frac{1}{2}\|e^{-iH_1t}(\,\cdot\,)e^{iH_1t}-e^{-iH_2t}(\,\cdot\,)e^{iH_2t}\|_\diamond\le \|e^{iH_1t}-e^{iH_2t}\|_{\opnorm}\le t\|H_1-H_2\|_{\opnorm}\,,
\end{equation*}
where the first inequality follows from \cite[Proposition 1.6]{haah2023query} and the second is a consequence of \cite[Lemma 3.2 (c)]{haah2023query}. This directly implies the bound in Eq.~\eqref{eq_upp_known}. Our novel contribution lies in the stronger inequality 
\begin{equation}\label{eq:sharper}
    \frac{1}{2}\|e^{-iH_1t}(\,\cdot\,)e^{iH_1t}-e^{-iH_2t}(\,\cdot\,)e^{iH_2t}\|_\diamond\le \sin \Big( \min \Big( \frac{\pi}{2}, t\Vert H_{1} - H_{2}\Vert_{\opnorm} \Big)\Big)\,,
\end{equation}
which directly implies Eq.~\eqref{eq_upp_not_known} and, to the best of our knowledge, does not appear in the existing literature. In the following, we provide a proof of this improved bound. Before proceeding, it is useful to recall the following result, given by Pfeifer \cite{pfeifer}.
\begin{theorem}[{\cite{pfeifer}}]
Let $H(t')$ be a time-dependent Hamiltonian, and let $\ket{\psi_t}$ denote the solution of the time-dependent Schrödinger equation $i\hbar \frac{d}{dt} \ket{\psi_t} = H(t) \ket{\psi_t}$, with initial state $\ket{\psi_{t=0}} \coloneqq  \ket{\psi_0}.$
Then, the overlap between the evolved state $\ket{\psi_t}$ and the initial state $\ket{\psi_0}$ satisfies the lower bound 
\begin{equation}\label{eq:pfeifer_thm}
    |\langle \psi_{0} | \psi_t \rangle| \ge \cos \Biggl( \frac{1}{\hbar} \int_0^t \sqrt{\bra{\psi_0}H(t')^2\ket{\psi_0} 
- \bra{\psi_0}H(t')\ket{\psi_0}^2} \, dt' \Biggr).
\end{equation}
\end{theorem}
We are now ready to prove Lemma~\ref{lem:upperboundtodT}.
\begin{proof}[Proof of Lemma ~\ref{lem:upperboundtodT}]
As recalled in Eq.~\eqref{eq:aharonov}, the diamond norm between two unitary channels can be expressed as
\begin{equation*}
    \frac{1}{2}\|e^{-iH_1t}(\,\cdot\,)e^{iH_1t}-e^{-iH_2t}(\,\cdot\,)e^{iH_2t}\|_\diamond = \frac{1}{2}\sup_{\ket{\psi}} \Vert e^{-iH_1t} \ketbra{\psi}{\psi} e^{iH_1t}-e^{-iH_2t}\ketbra{\psi}{\psi} e^{iH_2t}\Vert_{\trnorm}.
\end{equation*}
Hence, by using Eq.~\eqref{eq_diamond_pure1}, we have
\begin{equation}\label{eq:trace_dist_pure}
    \frac12\Vert e^{-iH_1t} \ketbra{\psi}{\psi} e^{iH_1t}-e^{-iH_2t}\ketbra{\psi}{\psi} e^{iH_2t}\Vert_{S_{1}} = \sqrt{1-|\bra{\psi}e^{-iH_{1}t} e^{iH_{2}t}\ket{\psi}|^{2}}\,.
\end{equation} 
To bound the overlap inside the square root, if we define $\ket{\psi(t)} \coloneqq e^{-iH_{1}t}e^{iH_{2}t}\ket{\psi}$, we can recognize the time-dependent Schr\"odinger equation 
\begin{equation*}
\frac{d\ket{\psi(t)}}{dt} = iH(t)\ket{\psi(t)}\,,
\end{equation*}
where $H(t) \coloneqq e^{-iH_{1}t}(H_{2}-H_{1})e^{iH_{1}t}$. In this way, we can bound the scalar product in Eq.~\eqref{eq:trace_dist_pure} by using the result from Pfeifer \cite{pfeifer} in Eq.~\eqref{eq:pfeifer_thm}
\begin{equation*}\label{eq:pfeifer}
| \langle \psi | \psi(t) \rangle | \ge \cos \Big(\min \Big( \frac{\pi}{2} , \int_{0}^{t}dt' \sqrt{\bra{\psi}H^{2}(t')\ket{\psi} - |\bra{\psi}H(t')\ket{\psi}|^2} \Big)\Big).
\end{equation*}
Hence, this gives

\begin{equation*}
\begin{split}
    \sup_{\ket{\psi}} \sqrt {1-| \langle \psi | \psi(t) \rangle |^{2} }
    &\le \sup_{\ket{\psi}} \sin \Bigg( \min \Bigg( \frac{\pi}{2}, \int_{0}^{t} dt'\,\sqrt{\bra{\psi}H^{2}(t')\ket{\psi}} \Bigg) \Bigg) \\
    &\le \sin \Bigg( \min \Bigg( \frac{\pi}{2}, \int_{0}^{t} dt'\,\|H(t')\|_{\opnorm} \Bigg) \Bigg),
\end{split}
\end{equation*}
where we have used $\Vert \Theta\Vert_{\opnorm} \coloneqq \sup_{\ket{\psi}} \sqrt{\bra{\psi}\Theta^
{\dagger}\Theta\ket{\psi}}$. Finally, we use the unitary invariance of the operator norm to simplify the expression inside the integral:
\begin{equation*}
    \|H(t')\|_{\opnorm} = \|H_2 - H_1\|_{\opnorm}.
\end{equation*}
Therefore,
\begin{equation*}
\sin \Bigg( \min \Bigg( \frac{\pi}{2}, \int_{0}^{t} dt'\,\|H(t')\|_{\opnorm} \Bigg) \Bigg)
= \sin \Big( \min \Big( \frac{\pi}{2}, t\|H_{1}-H_{2}\|_{\opnorm} \Big)\Big),
\end{equation*}
which, by putting everything together, proves the Lemma.
\end{proof}

\subsubsection{Lower bounding the time-constrained distance}\label{sec_sub2}
This subsection is devoted to derive a general lower bound on the time-constrained diamond distance between Hamiltonians in terms of the operator norm. The main result of this subsection is the following lemma, which will be proved at the end.
\begin{lemma}[Diamond distance lower bound on Hamiltonians]\label{lemma_2}
Let $H_1,H_2$ be Hamiltonians such that $\|H_1\|_{\opnorm},\|H_2\|_{\opnorm}\le 1$ and $\Tr H_1=\Tr H_2$. Let $T>0$. Then,
\begin{equation*}
        d_T(H_1,H_2)\;\ge\; \frac{1}{4\pi}\min\!\left(T,\frac{1}{4\pi}\right)\|H_1-H_2\|_{\opnorm}.
\end{equation*}
\end{lemma}
To start, we introduce the following notation:
\begin{align}\label{eq_mod}
    p(x) &\coloneqq x \!\!\!\!\!\mod 2\pi \qquad \forall\, x\in\mathbb{R}, \\
    q(x) &\coloneqq (x+\pi)\!\!\!\!\!\mod 2\pi - \pi \qquad \forall\, x\in\mathbb{R}.
\end{align}
Intuitively, $p(x)$ denotes the projection of $x$ onto the interval $[0,2\pi]$, while $q(x)$ denotes the projection of $x$ onto $[-\pi,\pi]$.  We now begin with the following auxiliary lemma, whose proof can be found in \cref{appendix:proof of lemma minmax}.

\begin{restatable}{lemma}{theoremx}
\label{lemma_minmax}
For any $a,b\in\mathbb{R}$, it holds that
\begin{equation*}\label{eq:minmax}
    \min_{x\in\mathbb{R}} \max\bigl(|q(a-x)|, |q(b-x)|\bigr) 
    = \frac{1}{2}\min\!\left(|q(a)-q(b)|,\,|p(a)-p(b)|\right).
\end{equation*}    
\end{restatable}

We now present the following lemma, which provides a lower bound on the diamond distance between a unitary Hamiltonian evolution and the identity channel.  
\begin{lemma}\label{lemma_diff}
    Let $H$ be a $d$-dimensional Hamiltonian with eigenvalues $\lambda_1,\ldots,\lambda_d$. Then
    \begin{equation*}
        \frac{1}{2}\left\|e^{-iH}(\cdot)e^{iH}-\Id\right\|_\diamond 
        \;\ge\; \frac{1}{2\pi}\max_{j,k\in[d]}\min\!\bigl(|p(\lambda_j)-p(\lambda_k)|,\,|q(\lambda_j)-q(\lambda_k)|\bigr),
    \end{equation*}
    where $p$ and $q$ are defined in Eq.~\eqref{eq_mod}.
\end{lemma}

\begin{proof}
    Fix $j,k\in[d]$. By Lemma~\ref{eq_dia_inf}, we have
    \begin{equation*}
        \|e^{-iH}(\cdot)e^{iH}-\Id\|_\diamond 
        \;\ge\; \min_{\phi\in\mathbb{R}}\|e^{i(H-\phi\Id)}-\Id\|_{\opnorm}.
    \end{equation*}
    Since the operator norm is at least as large as the modulus of any eigenvalue, it follows that
    \begin{equation*}
        \min_{\phi\in\mathbb{R}}\|e^{i(H-\phi\Id)}-\Id\|_{\opnorm}
        \;\ge\; \min_{\phi\in\mathbb{R}}\max\!\left(|e^{i(\lambda_j-\phi)}-1|,\,|e^{i(\lambda_k-\phi)}-1|\right).
    \end{equation*}
    For any $\theta\in\mathbb{R}$,
    \begin{equation*}\label{eq_xxxxxxxxxx}
        |e^{i\theta}-1| = |e^{iq(\theta)}-1| 
        = 2\Bigl|\sin\!\Bigl(\tfrac{q(\theta)}{2}\Bigr)\Bigr| 
        \;\ge\; \tfrac{2}{\pi}|q(\theta)|,
    \end{equation*}
    where the last inequality uses $|\sin(x/2)| \ge |x|/\pi$ for all $x\in[-\pi,\pi]$.  
    Combining the above gives
    \begin{equation*}
        \tfrac{1}{2}\|e^{-iH}(\cdot)e^{iH}-\Id\|_\diamond 
        \;\ge\; \tfrac{1}{\pi}\min_{\phi\in\mathbb{R}}  \max \!\left(|q(\lambda_j-\phi)|,|q(\lambda_k-\phi)|\right).
    \end{equation*}
    Finally, applying Lemma~\ref{lemma_minmax} yields
    \begin{equation*}
        \tfrac{1}{2}\|e^{-iH}(\cdot)e^{iH}-\Id\|_\diamond 
        \;\ge\; \tfrac{1}{2\pi}\min\!\left(|q(\lambda_j)-q(\lambda_k)|,\,|p(\lambda_j)-p(\lambda_k)|\right).
    \end{equation*}
    To finish, it suffices to take the supremum of the above equation over all $j,k\in [d].$
\end{proof}

As a consequence of \cref{lemma_diff}, we obtain the following. Given a Hamiltonian $H,$ we define $R(H)$ as $R(H):=\max |\lambda-\lambda'|$, where the maximum ranges over all eigenvalues $\lambda$ and $\lambda'$ of $H.$ 

\begin{lemma}\label{lemma_cons2}
    Let $H$ be a Hamiltonian with $\|H\|_{\opnorm}\le \tfrac{\pi}{2}$. Then,
    \begin{equation*}
        \tfrac{1}{2}\left\|e^{-iH}(\cdot)e^{iH}-\Id\right\|_\diamond 
        \;\ge\; \tfrac{1}{2\pi}R(H).
    \end{equation*}
\end{lemma}

\begin{proof}
    By Lemma~\ref{lemma_diff},
    \begin{equation*}
        \tfrac{1}{2}\|e^{-iH}(\cdot)e^{iH}-\Id\|_\diamond 
        \;\ge\; \tfrac{1}{2\pi}\max_{j,k}\min\!\bigl(|p(\lambda_j)-p(\lambda_k)|,\,|q(\lambda_j)-q(\lambda_k)|\bigr).
    \end{equation*}
    Since $\|H\|_{\opnorm}\le \tfrac{\pi}{2}$, all eigenvalues of $H$ lie in $[-\tfrac{\pi}{2},\tfrac{\pi}{2}]$. Thus, $q(\lambda)=\lambda$ for each eigenvalue $\lambda$, and we obtain
    \begin{equation*}
        \tfrac{1}{2}\|e^{-iH}(\cdot)e^{iH}-\Id\|_\diamond 
        \;\ge\; \tfrac{1}{2\pi}\max_{j,k}|\lambda_j-\lambda_k|
        = \tfrac{1}{2\pi}R(H).
    \end{equation*}
\end{proof}

\cref{lemma_cons2} yields the following result, which is almost the result we would like to prove, but with the caveat of the second Hamiltonian being equal to $0$.

\begin{lemma}\label{lemma_low_id}
    Let $H$ be a Hamiltonian with $\|H\|_{\opnorm}\le \tfrac{\pi}{2}$ and $\Tr H=0$. Then,
    \begin{equation*}
        \tfrac{1}{2}\left\|e^{-iH}(\cdot)e^{iH}-\Id\right\|_\diamond 
        \;\ge\; \tfrac{1}{2\pi}\|H\|_{\opnorm}.
    \end{equation*}
\end{lemma}

\begin{proof}
    Since $\Tr H=0$, $H$ must have both positive and negative eigenvalues, hence $R(H)\ge \|H\|_{\opnorm}$. Applying Lemma~\ref{lemma_cons2} completes the proof.
\end{proof}

We now address the case of non-commuting Hamiltonians.

\begin{lemma}\label{lemma_2pi}
    Let $H_1,H_2$ be Hamiltonians such that $\|H_1\|_{\opnorm},\|H_2\|_{\opnorm}\le \tfrac{1}{4\pi}$ and $\Tr H_1=\Tr H_2$. Then
    \begin{equation*}
         \frac{1}{2}\left\|e^{-iH_1}(\cdot)e^{iH_1}-e^{-iH_2}(\cdot)e^{iH_2}\right\|_\diamond
         \;\ge\; \frac{1}{4\pi}\|H_1-H_2\|_{\opnorm}.
    \end{equation*}
\end{lemma}

\begin{proof}
We proceed as follows:
\begin{align*}
    &\frac{1}{2}\left\|e^{-iH_1}(\cdot)e^{iH_1}-e^{-iH_2}(\cdot)e^{iH_2}\right\|_\diamond \notag \\
    &\quad \overset{\mathrm{(i)}}{=}\;
    \frac{1}{2}\left\|e^{-iH_1}e^{iH_2}(\cdot)e^{iH_1}e^{-iH_2}-\Id\right\|_\diamond \notag \\
    &\quad \overset{\mathrm{(ii)}}{\ge}\;
    \frac{1}{2}\left\|e^{-i(H_1-H_2)}(\cdot)e^{i(H_1-H_2)}-\Id\right\|_\diamond \notag  - \frac{1}{2}\left\|e^{-i(H_1-H_2)}(\cdot)e^{i(H_1-H_2)}
      -e^{-iH_1}e^{iH_2}(\cdot)e^{iH_1}e^{-iH_2}\right\|_\diamond \notag \\
    &\quad \overset{\mathrm{(iii)}}{\ge}\;
    \frac{1}{2\pi}\|H_1-H_2\|_{\opnorm}
    - \frac{1}{2}\left\|e^{-i(H_1-H_2)}(\cdot)e^{i(H_1-H_2)}
      -e^{-iH_1}e^{iH_2}(\cdot)e^{iH_1}e^{-iH_2}\right\|_\diamond \notag \\
    &\quad \overset{\mathrm{(iv)}}{\ge}\;
    \frac{1}{2\pi}\|H_1-H_2\|_{\opnorm}
    - \left\|e^{-i(H_1-H_2)}-e^{-iH_1}e^{iH_2}\right\|_{\opnorm} \notag \\
    &\quad \overset{\mathrm{(v)}}{\ge}\;
    \frac{1}{2\pi}\|H_1-H_2\|_{\opnorm}
    - \tfrac{1}{2}\|[H_1,H_2]\|_{\opnorm} \notag \\
    &\quad \overset{\mathrm{(vi)}}{\ge}\;
    \frac{1}{2\pi}\|H_1-H_2\|_{\opnorm}
    - \tfrac{1}{4\pi}\|H_1-H_2\|_{\opnorm} \notag \\
    &\quad = \frac{1}{4\pi}\|H_1-H_2\|_{\opnorm}\,.
\end{align*}

The justifications are as follows:  
(i) invariance of the diamond norm under conjugation by unitaries;  
(ii) triangle inequality;  
(iii) application of Lemma~\ref{lemma_low_id} and the bounds on $\|H_1\|_{\opnorm},\|H_2\|_{\opnorm}$;  
(iv) Lemma~\ref{eq_dia_inf}, relating the diamond and operator norms;  
(v) Lemma~\ref{lemma_commm};  
(vi) the estimate
\begin{equation*}
    \|[H_1,H_2]\|_{\opnorm}
    \le \|H_1(H_2-H_1)+(H_1-H_2)H_1\|_{\opnorm}
    \le 2\|H_1\|_{\opnorm}\|H_1-H_2\|_{\opnorm},
\end{equation*}
together with $\|H_1\|_{\opnorm}\le \tfrac{1}{4\pi}$. This concludes the proof.
\end{proof}

We can now prove the main result of this section.

\begin{proof}[Proof of \cref{lemma_2}]
    By definition of the time-constrained distance and \cref{lemma_2pi},
    \begin{align*}
        d_T(H_1,H_2)
        &= \sup_{t\in[0,T]}\,\frac{1}{2}\left\|e^{-itH_1}(\cdot)e^{itH_1}-e^{-itH_2}(\cdot)e^{itH_2}\right\|_\diamond \\
        &\ge \sup_{t\in[0,\min(T,\tfrac{1}{4\pi})]}\,\frac{1}{2}\left\|e^{-itH_1}(\cdot)e^{itH_1}-e^{-itH_2}(\cdot)e^{itH_2}\right\|_\diamond \\
        &\ge \sup_{t\in[0,\min(T,\tfrac{1}{4\pi})]}\,\frac{1}{4\pi}\|tH_1-tH_2\|_{\opnorm} \\
        &= \frac{1}{4\pi}\min\!\left(T,\frac{1}{4\pi}\right)\|H_1-H_2\|_{\opnorm}.
    \end{align*}
\end{proof}
\begin{remark}
The constants in Lemma~\ref{lemma_2} can be slightly improved by refining the inequality $|\sin(x/2)|\ge \tfrac{1}{\pi}|x|$ (valid for $x\in[-\pi,\pi]$) used in Lemma~\ref{lemma_diff}. In particular, if the Hamiltonians satisfy $\|H\|_{\opnorm}\le c$ for some $c<\pi$, as in Lemma~\ref{lemma_2pi}, sharper constants can be obtained.
\end{remark}

\subsection{Temperature-constrained trace distance}
In statistical mechanics and quantum thermodynamics, it is often of interest to compare the equilibrium properties of systems governed by different Hamiltonians. A natural way to do this is to compare their Gibbs states, $\rho_i(\beta) = e^{-\beta H_i} / \mathrm{Tr}[e^{-\beta H_i}]$, which describe the state of the system governed by the Hamiltonians $H_i$ at inverse temperature $\beta$. Such a comparison can be used, for example, to quantify how a perturbation to the Hamiltonian affects the thermal properties of the system, or to assess the accuracy of a simulation of a thermal state.

From a physical standpoint, it is reasonable to restrict this comparison to a finite range of temperatures, corresponding to a bounded interval $\beta \in [0,B]$ of accessible inverse temperatures. In realistic experiments, one cannot prepare systems at arbitrarily low temperatures (large $\beta$) due to practical limitations or constraints on preparation times. Focusing on a finite range of $\beta$ therefore reflects actual laboratory capabilities.

From a mathematical perspective, thanks to the Holevo-Helstrom theorem, employing the trace distance to compare Gibbs states provides an operational meaning: 
$\tfrac{1}{2}\|\rho_1(\beta) - \rho_2(\beta)\|_{\trnorm}$ quantifies the optimal distinguishability of the two thermal states at inverse temperature $\beta$. Taking the supremum over $\beta \in [0,B]$ then corresponds to the maximum distinguishability achievable within the permitted temperature range.

Motivated by these considerations, we introduce the \textit{temperature-constrained trace distance} between Hamiltonians:

\begin{equation}\label{eq:temp_constr_diamond}
    d_B(H_{1},H_{2}) \coloneqq \frac{1}{2}\sup_{\beta\in[0,B]} \Vert \rho_1(\beta)-\rho_2(\beta) \Vert_{\trnorm},
\end{equation}
where $\rho_i(\beta)\coloneqq e^{-\beta H_i}/\Tr[e^{-\beta H_i}]$. 

\begin{remark}
It is instructive to contrast the definition of the temperature-constrained trace distance in Eq.~\eqref{eq:temp_constr_diamond} with the unconstrained setting, i.e.~when $B=\infty$. In this regime one has
\begin{equation*}
    d_{\infty}(H_1,H_2)\;\geq\; \tfrac{1}{2}\,\|\rho_1(\infty)-\rho_2(\infty)\|_{\trnorm}.
\end{equation*}
If a Hamiltonian $H$ has a unique ground state, then the corresponding Gibbs state at infinite inverse temperature reduces to this ground state.  
For instance, for $H_1=\eps Z$ and $H_2=-\eps Z$, we obtain $\rho_1(\infty)=\ket{0}\!\bra{0}$ and $\rho_2(\infty)=\ket{1}\!\bra{1}$, which gives
\begin{equation*}
    d_{\infty}(H_1,H_2)=1,
\end{equation*}
even though $H_1$ and $H_2$ are arbitrarily close (e.g.~in operator norm).
This highlights that, just as in the time-unconstrained setting, allowing arbitrarily low temperatures leads to maximal distinguishability, even between nearly identical Hamiltonians.
\end{remark}

Similarly to what we have done for the time-constrained diamond distance, we can upper bound $d_B$ in terms of a operator norm between the Hamiltonian.

\subsubsection{Upper bounding the temperature-constrained distance}\label{sec_sub3}
This subsection is devoted to derive an upper bound on the temperature-constrained trace distance in terms of the operator norm of the Hamiltonian difference. We start with the following lemma, which provides un upper bound on the trace distance between Gibbs states. The strategy is to relate the trace distance between Gibbs states to the symmetrized Kullback–Leibler divergence (SKL) via Pinsker’s inequality, and then obtain a bound on the SKL in terms of $\|H_1-H_2\|_{\opnorm}$.
\begin{lemma}\label{new_leee}
For all Hamiltonians $H_1,H_2$ it holds that:
\begin{equation}\label{eq:mainbound}
    \left\|\frac{e^{H_1}}{\Tr e^{H_1}}-\frac{e^{H_2}}{\Tr e^{H_2}}\right\|_{\trnorm}
    \;\leq\;\|H_1-H_2\|_{\opnorm}.
\end{equation}
\end{lemma}
\cref{new_leee} is strictly tighter than the known bound of \cite[Lemma 16]{brandao2017quantumspeedupssemidefiniteprogramming}, which states that
\begin{equation*}
    \left\|\frac{e^{H_1}}{\Tr e^{H_1}}-\frac{e^{H_2}}{\Tr e^{H_2}}\right\|_{\trnorm}\le 2\left(e^{\|H_1-H_2\|_{\opnorm}}-1\right)\,.
\end{equation*}
\begin{proof}[Proof of \cref{new_leee}]
    Let $\rho_1\coloneqq \frac{e^{H_1}}{\Tr e^{H_1}}$ and $\rho_2\coloneqq \frac{e^{H_2}}{\Tr e^{H_2}}$.     The symmetrized Kullback–Leibler divergence (SKL)~\cite{kullback1951} is defined as
    \begin{equation*}
        d_{\mathrm{SKL}}(\rho_1,\rho_2)
        \coloneqq D(\rho_1\|\rho_2)+D(\rho_2\|\rho_1),
    \end{equation*}
    where $D(\rho_1\|\rho_2)$ is the relative entropy between the states $\rho_1$ and $\rho_2$, defined as
    \begin{equation*}
        D(\rho_1\|\rho_2) \coloneqq \Tr \left[\rho_1\,(\ln \rho_1 - \ln \rho_2)\right]\,.
    \end{equation*}
    Our first claim is that the following inequality holds:
    \begin{equation}\label{eq:SKLbound}
        d_{\mathrm{SKL}}(\rho_1,\rho_2)
        \;\leq\,\|H_1-H_2\|_{\opnorm}\,\|\rho_1-\rho_2\|_{\trnorm}.
    \end{equation}
    To prove \cref{eq:SKLbound}, we start by computing $d_{\mathrm{SKL}}$ explicitly for Gibbs states. Using the definition of relative entropy, we have
    \begin{equation}\label{eq:d_skl_1}
        d_{\mathrm{SKL}}(\rho_1,\rho_2)
        = \Tr[\rho_1(\ln\rho_1-\ln\rho_2)]
          +\Tr[\rho_2(\ln\rho_2-\ln\rho_1)].
    \end{equation}
    Hence, by substituting in Eq.~\eqref{eq:d_skl_1} the explicit expression of the Gibbs states, we have
    \begin{align*}
        d_{\mathrm{SKL}}(\rho_1,\rho_2)
        &= -\Tr[\rho_1 H_1]-\ln Z_1 
           +\Tr[\rho_1 H_2]+\ln Z_2\\
        &\quad -\Tr[\rho_2 H_2]-\ln Z_2
           +\Tr[\rho_2 H_1]+\ln Z_1\\
        &= \Tr\big[(\rho_1 -\rho_2) (H_2-H_1)\big].
    \end{align*}  
    By exploiting Hölder's inequality, this expression immediately implies \cref{eq:SKLbound}. Given this, we now want to prove Eq.~\eqref{eq:mainbound}. To this end, we recall Pinsker’s inequality \cite{pinsker}, which relates relative entropy to trace distance:
    \begin{equation*}
        \frac{1}{2}\|\rho_1-\rho_2\|_{\trnorm}^2\le D(\rho_1\|\rho_2)\,.
    \end{equation*}
    Consequently, 
\begin{align*}
    \|\rho_1-\rho_2\|_{\trnorm}^2&= \frac{1}{2}\|\rho_1-\rho_2\|_{\trnorm}^2+\frac{1}{2}\|\rho_2-\rho_1\|_{\trnorm}^2\leq D(\rho_1\|\rho_2)+D(\rho_2\|\rho_1)=d_{\mathrm{SKL}}(\rho_1,\rho_2).
\end{align*}
Combining this with Eq.~\eqref{eq:SKLbound}, we arrive at
    \begin{equation}
        \|\rho_1 -\rho_2 \|_{\trnorm}^2
        \;\leq\,\|H_1-H_2\|_{\opnorm}\,\|\rho_1-\rho_2\|_{\trnorm}\,,
    \end{equation}
    which proves the lemma.
    \end{proof}

Note that if $H_1 \mapsto -\beta H_1$ and $H_2 \mapsto -\beta H_2$, by taking the supremum over $\beta\in [0,B]$ we can prove the following result
\begin{lemma}\label{lemma_dB_upperbound}
    Let $H_1,H_2$ be Hamiltonians and let $d_B(H_1,H_2)$ be the temperature-constrained trace distance between $H_1$ and $H_2$. Then
    \begin{equation*}\label{eq:d_B_bound}
         d_B(H_1,H_2)\;\leq\;\frac{B}{2}\,\|H_1-H_2\|_{\opnorm}\,.
    \end{equation*}
\end{lemma}

\subsubsection{No lower bound for the temperature-constrained distance}\label{sec:nolowboundfordB}
We now investigate whether the distance $d_B(H_1,H_2)$ admits a nontrivial lower bound in terms of the operator norm of the difference of Hamiltonians. In particular, one might ask whether inequalities of the form
\begin{equation*}
    d_B(H_1,H_2)\geq\poly\bigl(B,\|H_1-H_2\|_{\opnorm}\bigr)
\end{equation*}
can hold in general. The following example shows that this is not the case: even when the operator norm of the difference between $H_1$ and $H_2$ is constant, the corresponding Gibbs states can become arbitrarily close in trace distance as the system size increases. To see this, consider the $n$-qubit Hamiltonians
\begin{equation*}
    H_1=\ket{0^n}\bra{0^n}-\ket{1^n}\bra{1^n},
    \qquad
    H_2=-\ket{0^n}\bra{0^n}+\ket{1^n}\bra{1^n}.
\end{equation*}
In the computational basis their spectra are $\spec(H_1)=\{+1,-1,0,\dots,0\}$ and $\spec(H_2)=\{-1,+1,0,\dots,0\}$, where $0$ appears $2^n-2$ times. Consequently, $e^{-\beta H_1}$ and $e^{-\beta H_2}$ are diagonal and the partition functions coincide, i.e. $Z_1(\beta)=Z_2(\beta)=\Tr(e^{-\beta H_1})=\Tr(e^{-\beta H_2})=2^n-2+e^{\beta}+e^{-\beta}.$ Thus, the corresponding Gibbs states at inverse temperature $\beta$ are
\begin{equation*}
\rho_1(\beta)=\frac{1}{Z_1(\beta)}\begin{pmatrix}
    e^{-\beta} & & &\\
    & 1 & & \\
    & & \ddots &\\
    & & & e^{\beta}
\end{pmatrix}, \qquad \rho_2(\beta)=\frac{1}{Z_2(\beta)}\begin{pmatrix}
    e^{\beta} & & &\\
    & 1 & & \\
    & & \ddots &\\
    & & & e^{-\beta}
\end{pmatrix}.
\end{equation*}
Therefore, the trace norm of their difference is
\begin{equation}\label{eq:rho1_rho2_norm1}
    \|\rho_1(\beta)-\rho_2(\beta)\|_{\trnorm} = \frac{2(e^\beta - e^{-\beta})}{2^n - 2 + e^\beta + e^{-\beta}}.
\end{equation}
Fixing $B$ as a constant, we have that for any $\beta \in [0,B]$ the numerator of Eq.~\eqref{eq:rho1_rho2_norm1} remains of order $O(1)$ while the denominator grows as $2^n$. Hence,
\begin{equation*}
    \sup_{\beta\in[0,B]}\|\rho_1(\beta)-\rho_2(\beta)\|_{\opnorm} = O\left(\frac{1}{2^n}\right) \to 0 \quad \text{as } n \to \infty.
\end{equation*}

\bibliographystyle{alpha}
\bibliography{notes}

\appendix
\section{Proof of Lemma \ref{lemma_minmax}}\label{appendix:proof of lemma minmax}
In this appendix, we prove Lemma~\ref{lemma_minmax}, which we restate for reader's convenience.
\theoremx*

\begin{proof}
We first show that the minimum of
\begin{equation}
    \min_{x \in \mathbb{R}} \max\bigl(|q(a-x)|, |q(b-x)|\bigr)
\end{equation}
is attained at solutions of
\begin{equation}\label{eq:balance}
    |q(a-x)| = |q(b-x)| \,.
\end{equation}
To begin, consider the following general observation. Given continuous functions $f,g:\mathbb{R}\to\mathbb{R}$, the minimizers of
\begin{equation}
    x \mapsto \max(f(x),g(x))
\end{equation}
must lie in one of the following sets:
\begin{enumerate}
    \item[(i)] Points where $f$ has a local minimum and $f(x)>g(x)$,
    \item[(ii)] Points where $g$ has a local minimum and $g(x)>f(x)$,
    \item[(iii)] Points where $f(x)=g(x)$.
\end{enumerate}
In our case, with
\begin{equation}
    f(x) \coloneqq |q(a-x)|, \qquad g(x) \coloneqq |q(b-x)|,    
\end{equation}
each local minimum of $f$ or $g$ equals $0$ and $f$ and $g$ are nonnegative. Thus, (i) and (ii) do not occur, and minimizers must belong to case (iii), i.e. satisfy \eqref{eq:balance}.

We now analyse the solutions of Eq.~\eqref{eq:balance}. Such solutions must satisfy one of the following two conditions:
\begin{enumerate}
    \item $q(b-x)=q(a-x)$, which is true only if $q(a)=q(b)$, in which case both sides of \cref{eq:minmax} equal 0, so the statement is true.
    \item $q(b-x)=-q(a-x)$, which gives $q(2x)=q(a+b)$, i.e.

    \[
        x = \tfrac{a+b}{2} + n\pi, \qquad n\in\mathbb{Z}.
    \]
\end{enumerate}
Hence, since the function $x\mapsto \max(|q(a-x)|,|q(b-x)|)$ is $2\pi$-periodic, the minimum is achieved at $x=\tfrac{a+b}{2}$ or $x=\tfrac{a+b}{2}+\pi$. Hence, we obtain that
\begin{align}
    \min_{x\in\mathbb{R}} \max(|q(a-x)|,|q(b-x)|) 
    &= \min_{x\in\mathbb{R}} \max(|q(q(a)-x)|,|q(q(b)-x)|) \\
    &= \min\!\left(\biggl|q\Bigl(\tfrac{q(a)-q(b)}{2}\Bigr)\biggr|, \, 
       \biggl|q\Bigl(\tfrac{q(a)-q(b)}{2}-\pi\Bigr)\biggr|\right) \\
    &= \min\!\left(\tfrac{|q(a)-q(b)|}{2},\, \Bigl|q\Bigl(\tfrac{q(b)+2\pi-q(a)}{2}\Bigr)\Bigr|\right).\label{eq:45}
\end{align}
Without loss of generality, we assume $q(a)\ge q(b)$. Two cases arise:
\paragraph{Case 1: $\tfrac{q(a)-q(b)}{2}\in[0,\pi/2]$.}  
Then, $\bigl|q(\tfrac{q(b)+2\pi-q(a)}{2})\bigr|\in[\pi/2,\pi]$, so, by \cref{eq:45},
\begin{equation}\label{eq_case1}
    \min_{x\in\mathbb{R}} \max(|q(a-x)|,|q(b-x)|) 
    = \tfrac{q(a)-q(b)}{2} = \tfrac{|q(a)-q(b)|}{2}.
\end{equation}
Since $|q(a)-q(b)|\in[0,\pi]$, we also have
\begin{equation}\label{eq:48}
    |q(a)-q(b)| = \min(|q(a)-q(b)|,|p(a)-p(b)|).
\end{equation}
Combining \cref{eq_case1,eq:48} yields the desired result.

\paragraph{Case 2: $\tfrac{q(a)-q(b)}{2}\in[\pi/2,\pi]$.}  
Then, $q(a)\in[0,\pi]$ and $q(b)\in[-\pi,0]$, so
\begin{equation}
    q(a)=p(a), \qquad q(b)+2\pi=p(b).
\end{equation}
Thus, as $p(a),\, p(b) \in [0,2\pi]$, it follows that $(p(a)-p(b))/2\in [-\pi,\pi]$, so we have that
\begin{equation}\label{eq:49}
    \biggl|q\Bigl(\tfrac{q(b)+2\pi-q(a)}{2}\Bigr)\biggr|
    = \left|q\left(\frac{p(b)-p(a)}{2}\right)\right|=\frac{|p(b)-p(a)|}{2}.
\end{equation}
Now, combining \cref{eq:45,eq:49} yields the stated result.
\end{proof}

\end{document}